\documentclass[aps,pra,twocolumn,superscriptaddress,nofootinbib,notitlepage]{revtex4-1}

\usepackage{graphicx}
\usepackage{caption}
\usepackage{latexsym}
\usepackage{amsmath}
\usepackage{amssymb}
\usepackage{amsfonts}
\usepackage{color}
\usepackage{pgf}
\usepackage{amsthm}
\usepackage{dsfont}
\usepackage{tikz}
\usetikzlibrary{matrix}

\usepackage[utf8]{inputenc}
\usepackage[english]{babel}
\usepackage{ulem}
\usepackage{mwe} 
\usepackage{subfig}
\usepackage{mathtools}
\usepackage{amssymb}

\begin{document}

\newcommand{\ket}[1]{\vert #1 \rangle}
\newcommand{\bra} [1] {\langle #1 \vert}
\newcommand{\braket}[2]{\langle #1 | #2 \rangle}
\newcommand{\ketbra}[2]{| #1 \rangle \langle #2 |}
\newcommand{\proj}[1]{\ket{#1}\bra{#1}}
\newcommand{\mean}[1]{\langle #1 \rangle}
\newcommand{\opnorm}[1]{|\!|\!|#1|\!|\!|_2}

\newtheoremstyle{break}
  {\topsep}{\topsep}%
  {\itshape}{}%
  {\bfseries}{}%
  {\newline}{}%
\theoremstyle{break}
\newtheorem{theorem}{Theorem}
\newtheorem{lem}{Lemma}
\newtheorem{rem}{Remark}
\newtheorem{defin}{Definition}
\newtheorem{corollary}{Corollary}
 \newtheorem{conj}{Conjecture}
 \newtheorem*{prop}{Properties}
 \newcommand{\kket}[1]{\vert\vert #1 \rangle\rangle}
 \newcommand{\bbra} [1] {\langle \langle \rangle#1 \vert\vert}
\newcommand{\mmean}[1]{\langle\langle #1 \rangle\rangle}
\newcommand{\tr}{\mathrm{Tr}}
\newcommand{\red}[1]{\textcolor{red}{#1	}}
\newcommand{\blue}[1]{\textcolor{blue}{#1	}}

\title{Realignment separability criterion assisted with filtration\\ for detecting continuous-variable entanglement}

\author{ Anaelle Hertz}
\address{
	Department of Physics, University of Toronto, Toronto, Ontario M5S 1A7, Canada
}
\affiliation{Centre for Quantum Information and Communication, \'Ecole polytechnique de Bruxelles, CP 165, Universit\'e libre de Bruxelles, 1050 Brussels, Belgium}

\author{Matthieu Arnhem}
\affiliation{Centre for Quantum Information and Communication, \'Ecole polytechnique de Bruxelles, CP 165, Universit\'e libre de Bruxelles, 1050 Brussels, Belgium}

\author{Ali Asadian}
\affiliation{Department of Physics, Institute for Advanced Studies in Basic Sciences (IASBS), Gava Zang, Zanjan 45137-66731, Iran}
\affiliation{Vienna Center for Quantum Science and Technology, Atominstitut, TU Wien, 1040 Vienna, Austria}

\author{Nicolas J. Cerf}
\affiliation{Centre for Quantum Information and Communication, \'Ecole polytechnique de Bruxelles, CP 165, Universit\'e libre de Bruxelles, 1050 Brussels, Belgium}

\begin{abstract}
We introduce a weak form of the realignment separability criterion which is particularly suited to detect continuous-variable entanglement and is physically implementable (it requires linear optics transformations and homodyne detection). Moreover, we define a family of states, called Schmidt-symmetric states, for which 
the weak realignment criterion reduces to the original formulation of the realignment criterion, making it even more valuable as it is easily computable especially in higher dimensions. Then, we focus in particular on Gaussian states and introduce a filtration procedure based on noiseless amplification or attenuation, which enhances the entanglement detection sensitivity. In some specific examples, it does even better than the original realignment criterion.
\end{abstract}

\maketitle

\nopagebreak


\section{Introduction}\label{introduction}

When it comes to mixed states, determining whether a state is entangled or not is provably a hard decision problem \cite{horodecki,Guhne2009}. Still, it has long been and it remains an active research topic  because entanglement is a key resource for quantum information processing. Both for discrete- and continuous-variable systems, various separability criteria --- conditions that must be satisfied by any separable state --- have been derived. Probably the best known criterion is the Peres–Horodecki criterion \cite{peres, horodecki1996}, also called the positive partial transpose (PPT) criterion. Introduced for discrete-variable systems, it states that if a quantum state is separable, then its partial transpose must remain physical (i.e., positive semidefinite). This PPT condition is, in general, only a necessary condition for separability. It becomes sufficient only for systems of dimensions $2 \times 2$ and $2 \times 3$ \cite{horodecki1996}. The PPT criterion was generalized to continuous variables (i.e., infinite-dimensional systems) by Duan {\it et al.} \cite{duan} and Simon \cite{00Simon}. Interestingly, it is necessary and sufficient for all $1\times n$ Gaussian states \cite{WernerWolf} and $n \times m$ bisymmetric Gaussian states \cite{Serafini}. In all other cases, when a state is entangled but its partial transpose remains positive semidefinite, we call it a bound entangled state \cite{horodecki98,horodecki}. These are entangled states from which no pure entangled state can be distilled through local (quantum) operations and classical communications (LOCC).

Many other separability criteria have been developed over years (see, e.g., \cite{shchukin,walborn, Lami,Mardani,Mihaescu}, and consult \cite{horodecki} for an older, but still relevant, review). Among them, we focus in the present paper on the realignment criterion \cite{Chen,Rudolph}. This criterion is unrelated to the PPT criterion and thereby enables the detection of some bound entangled states in both discrete-variable \cite{Chen} and continuous-variable cases \cite{Zhang}. Unfortunately, the realignment criterion happens to be generally hard to compute, especially for continuous-variable systems. To our knowledge, it has only been computed for Gaussian states by Zhang {\it et al.} \cite{Zhang} and yet,  the difficulty increases with the number of modes.

In this paper, we introduce a weaker form of the realignment criterion which is much simpler to compute and comes with a physical implementation in terms of linear optics and homodyne detection, hence it is especially suited to detect continuous-variable entanglement. It is, in general, less sensitive to entanglement than the original realignment criterion and cannot detect bound-entangled states, but it happens to be equivalent to the original realignment criterion for the class of Schmidt-symmetric states. Furthermore, we show that by supplementing this criterion with a filtration method, it is possible to greatly improve it and sometimes even surpass the original realignment criterion while keeping the simplicity of computation.

In Sec. II, we review the definition of the realignment criterion, focusing especially on the realignment map $R$. We link different formulations of this criterion and list its main properties. In Sec. III, we introduce the weak realignment criterion and show that for a class of states that we call Schmidt-symmetric, both the weak and original realignment criteria are equivalent (while the former is much easier to compute than the latter). In Sect. IV, we apply the weak realignment criterion to continuous-variable states and give  special attention to Gaussian states. The idea is to compare to the work of Zhang {\it et al.} \cite{Zhang}, which relied on the original formulation of the criterion. We notice that several entangled states remain undetected by the weak realignment criterion and, unfortunately, the latter cannot detect bound entanglement. As a solution, we introduce in Sec. V a filtration procedure that enables a better entanglement detection by bringing the state closer to a Schmidt-symmetric state, hence increasing the sensitivity of the entanglement witness.  In Sec. VI, we provide some specific examples for $1\times 1$ and $2 \times 2$ Gaussian states. In some cases, the filtration procedure supplementing the weak realignment criterion enables a better entanglement detection than the original realignment criterion. Finally, we give our conclusions in Sec. VII.

\section{Realignment criterion and realignment map}\label{sectionrealignment}

It is well known that any bipartite pure state  $\ket{\psi}_{AB}$ can be decomposed according to the  Schmidt decomposition $\vert \psi \rangle_{AB} = ~\sum_i \lambda_i \, \vert i_A \rangle \vert i_B \rangle,$  where $\vert i_A \rangle$ and $\vert i_B \rangle$ form orthonormal bases of subsystems $A$ and $B$, and the $\lambda_i$'s are non-negative real numbers satisfying $\sum_i \lambda_i^2 = 1$ known as the Schmidt coefficients \cite{NielsenChuang}. The number of nonzero coefficients is called the Schmidt  rank and denoted as $r$. A pure state is entangled if and only if $r>1$. Interestingly, the entanglement classes under LOCC transformations are uniquely determined by the Schmidt rank \cite{Nielsen}.

An analogous Schmidt decomposition can also be defined for mixed states \cite{Peres93}.
Let $\rho$ be a mixed quantum state of a bipartite system \textit{AB}, then it can be written in its operator Schmidt decomposition as
\begin{equation}
\rho = \sum_{i=1}^{r} \lambda_i \, A_i \otimes B_i,
\label{eq-Schmidt}
\end{equation}
with the Schmidt coefficients $\lambda_i$ being some non-negative real numbers, the Schmidt rank $r$ satisfying $1\leq r\leq \min\{\dim A,\dim B\}$, and with $\{A_i\}$ and $\{B_i\}$ forming orthonormal bases\footnote{If the operator is Hermitian (such as $\rho$), then the operators $A_i$ and $B_i$ can be chosen Hermitian too. But the Schmidt decomposition is not unique and there exist other possible Schmidt decompositions of an Hermitian operator with non-Hermitian operators $A_i$ and $B_i$.}  of the operator spaces for subsystems $A$ and $B$ with respect to the Hilbert-Schmidt inner product, i.e., $\tr(A_i^\dag A_j)=\tr(B^\dag_iB_j)=~\delta_{ij}$. The Schmidt coefficients $\lambda_i$ are unique for a bipartite state $\rho$ and reveal some of its characteristic features. For example, the purity of $\rho$ can be expressed as $\tr \,\rho^2 = \sum_{i=1}^r \lambda_i^2$.

Similarly as for pure states, the operator Schmidt decomposition can be employed as an entanglement criterion for mixed bipartite states; this is called the computable cross norm criterion  and is defined as follows.
\begin{theorem}[\textbf{Computable cross norm criterion} \cite{Rudolph}]
	\label{theoreal1}
Let $\rho$ be a state with the operator Schmidt decomposition $\rho = \sum_{i=1}^{r} \lambda_i A_i \otimes B_i$. If $\rho$ is separable, then $\sum_{i=1}^{r}  \lambda_i \leq 1$. Conversely, if $\sum_{i=1}^{r}  \lambda_i > 1$, then $\rho$ is entangled.
\end{theorem}
The proof is given in Appendix \ref{Appendix0} for completeness.

\medskip
There exists an alternative formulation of the computable cross norm criterion which, as we will see,  turns out to be more convenient when considering continuous-variable states. This reformulation is done by defining a linear map $R$ called \textit{realignment map}, whose action on the tensor product of matrices $A=\sum_{ij} a_{ij}\ketbra{i}{j}$ and $B=\sum_{kl} b_{kl}\ketbra{k}{l}$ is
\begin{equation}\label{realignmentmap}
R\big(A\otimes B\big)
=\sum_{ijkl} a_{ij}b_{kl}\ket{i}\ket{j}\bra{k}\bra{l}.
\end{equation}
Since, any bipartite state $\rho$ can be decomposed into  $A\otimes B$ products according to Eq.~\eqref{eq-Schmidt}, one can easily express its realignment $R(\rho)$ based on definition \eqref{realignmentmap}.
Thus, the realignment map simply interchanges the bra-vector $\bra{j}$ of the first subsystem with the ket-vector $\ket{k}$ of the second subsystem.
Note that the map $R$ is basis-dependent, namely, it depends on the basis in which the matrix elements $a_{ij}$ and $b_{kl}$  are expressed.
When applying $R$ to continuous-variable states in Secs. \ref{sect-realignement-continuous-variable}, \ref{sectionstep1} and \ref{Section:Examples}, we will always assume that $\ket{i}$, $\ket{j}$, $\ket{k}$, and $\ket{l}$ are Fock states, so that Eq.~\eqref{realignmentmap} must be understood in the Fock basis.

Using the state-operator correspondence implied by the Choi-Jamiolkowski isomorphism \cite{Choi,Jamio}, we can identify matrices with vectors living in the tensor-product ket space, namely $\ket{A}{=}\sum_{ij} a_{ij}\ket{i}\ket{j}$ and $\ket{B}{=}\sum_{kl} b_{kl}\ket{k}\ket{l}$. Their corresponding dual vectors are noted $\bra{A}=\sum_{ij} a^*_{ij}\bra{i}\bra{j}$ and $\bra{B}=\sum_{kl} b^*_{kl}\bra{k}\bra{l}$, living in the tensor-product bra space. Hence, the above map can be reexpressed as
\begin{equation}\label{realignmentmap2}
R\big(\,A\otimes B\,\big)= \ket{A} \bra{B^*} ,
\end{equation}
where complex conjugation is also applied in the preferred basis. 
Using the fact that\footnote{$( A\otimes\mathds{1})\ket{\Omega}=\sum_{ij}a_{ij}(\ket{i}\bra{j}\otimes\mathds{1})\sum_k\ket{k}\ket{k}=\sum_{ij}a_{ij}\ket{i}\ket{j}$}
\begin{eqnarray}
\ket{A}&=&\sum_{ij}a_{ij}\ket{i}\ket{j}=(A\otimes\mathds{1})\ket{\Omega} \nonumber\\
\text{and}\qquad\qquad&&\nonumber\\
 \bra{B^*}&=&\sum_{ij}b_{ij}\bra{i}\bra{j}=\bra{\Omega}( B^T\otimes\mathds{1}),
\end{eqnarray} 
where $\ket{\Omega}=\sum_{i}\ket{i}\ket{i}$ is the (unnormalized\footnote{ This definition of $\ket{\Omega}$ remains useful even for continuous-variable (infinite-dimensional) systems, where it can be interpreted as a (unnormalized) two-mode squeezed vacuum state with infinite squeezing.
The definition of $R$ given by Eq. \eqref{RmapwithEPR} remains thus valid with  $\ket{\Omega}=\sum_{i=0}^\infty\ket{ii}$, where $\ket{i}$ stand for Fock states.}) maximally entangled state and $\mathds{1}=\sum_i \ketbra{i}{i}$ is the identity matrix, one can also rewrite the realignment map as 
\begin{eqnarray}
R(A \otimes B) 
&=& ( A \otimes \mathds{1}) \ketbra{\Omega}{\Omega} ( B^T \otimes \mathds{1})\nonumber\\
&=&( A \otimes \mathds{1}) \ketbra{\Omega}{\Omega} ( \mathds{1} \otimes B).
\label{RmapwithEPR}
\end{eqnarray}
which will happen to be useful when considering the optical realization of the separability criterion.

It is obvious that $R(R(\rho))=\rho$, so that definition \eqref{realignmentmap2} can also be restated as
\begin{equation}
R\big(\, \ket{A} \bra{B} \,\big)= A\otimes B^* .
\end{equation}
Note the special cases
\begin{eqnarray}
 R(\mathds{1}\otimes \mathds{1}) &=& \ketbra{\Omega}{\Omega}   , \nonumber \\
 R(\ket{\Omega}\bra{\Omega}) &=&  \mathds{1}\otimes\mathds{1}  ,
 \label{eq-special-cases-of-R}
\end{eqnarray}
which are trivial consequences of $\ket{\mathds{1}}=\ket{\Omega}$ and $ \hat \Omega=\mathds{1}$.

It will also be useful in the following to define the \textit{dual realignment map} $R^{\dagger}$, which is such that $\tr(\rho_1\,R(\rho_2))=\tr(R^\dag(\rho_1)\,\rho_2)$. Definitions \eqref{realignmentmap2} and \eqref{RmapwithEPR} translate into
\begin{eqnarray}
R^\dag(A \otimes B) &=& \ket{B^T} \bra{A^\dag} , \nonumber \\
&=& ( B^T \otimes \mathds{1}) \ketbra{\Omega}{\Omega} ( A \otimes \mathds{1})  ,  \nonumber\\
&=&( \mathds{1} \otimes B ) \ketbra{\Omega}{\Omega} ( A \otimes \mathds{1}) .
\end{eqnarray}

Coming back to the question of separability, let us now state the following theorem.
\begin{theorem}[\textbf{Realignment criterion} \cite{Chen}]
	\label{theoreal2}
If the bipartite state $\rho$ is separable, then $\parallel R(\rho) \parallel_{tr}\leq 1$. Conversely, if $\parallel R(\rho) \parallel_{tr}\, > 1$, then $\rho$ is entangled.
\end{theorem}

\begin{proof}
From Eq.~\eqref{realignmentmap2}, the realignment of a product state is given by
\begin{equation}
R(\rho_A\otimes\rho_B)=\ketbra{\rho_A}{\rho_B^*},
\end{equation}
and therefore,
\begin{equation}
\parallel R(\rho_A\otimes\rho_B) \parallel_{tr}=\tr\sqrt{\ketbra{\rho_A}{\rho_B^*}\rho_B^*\rangle\bra{\rho_A}}\leq1,
\end{equation}
where $\parallel\mathcal{O}\parallel_{tr}=\tr(\sqrt{\mathcal{O}\mathcal{O}^\dag})$ denotes the trace norm\footnote{The trace norm of $\mathcal{O}$ is equivalent to the sum of the singular values of $\mathcal{O}$, which are given by the square roots of the eigenvalues of $\mathcal{O}\mathcal{O}^\dag$. For an Hermitian operator, the trace norm is simply equal to the sum of the absolute values of the eigenvalues.} of an operator $\mathcal{O}$ and the inequality is found 
using the Hilbert-Schmidt inner product, $\bra{A}B\rangle=\tr(A^\dag B)$. The convexity of the trace norm implies that $\parallel R(\rho) \parallel_{tr}\leq 1$, for any separable state $\rho=\sum_i p_i\rho^A_i\otimes \rho^B_i$, with $p_i\ge 0$ and $\sum_i p_i=1$. 
\end{proof}

Theorem \ref{theoreal2} is called the realignment criterion as the detection of entanglement exploits the map $R$. But it is interesting to note that $\parallel R(\rho) \parallel_{tr}$ coincides with the sum of the Schmidt coefficients of $\rho$, so the realignment criterion is actually equivalent to Theorem 1 \cite{zhangzhang,johnston}.
Indeed, let $\rho$ be a state with the operator Schmidt decomposition $\rho = \sum_i^{r} \lambda_i \, A_i \otimes B_i$. Then, according to Eq.~(\ref{realignmentmap2}),
\begin{equation}
R(\rho)=\sum_i \lambda_i \, R(A_i\otimes B_i)=\sum_i \lambda_i \, \ket{A_i}\bra{B_i^*}
\end{equation}
and
\begin{eqnarray}
\parallel R(\rho)\parallel_{tr}&=&\tr\left[\sqrt{\sum_{i,j}\lambda_i\lambda_j\ket{A_i}\mean{B_i^*|B_j^*}\bra{A_j}}\right]\nonumber\\
&=&\tr\left[\sqrt{\sum_{i}\lambda_i^2\ket{A_i}\bra{A_i}}\right]\nonumber\\
&=&\tr\left[\sum_{i}|\lambda_i|\ket{A_i}\bra{A_i}\right]=\sum_i\lambda_i
\end{eqnarray}
since  $\mean{A_i|A_j}=\mean{B_i|B_j}=\delta_{ij}$. Theorem~\ref{theoreal2} is thus equivalent to Theorem~\ref{theoreal1}.

As a trivial example of Theorem~\ref{theoreal2}, let us consider two $d$-dimensional systems (with $d\ge 2$). The maximally mixed state $\rho=\mathds{1} \otimes\mathds{1}/d^2$ is mapped to $R(\rho)= \ketbra{\Omega}{\Omega}/d^2$, see Eq. \eqref{eq-special-cases-of-R}, so its trace norm is $\parallel~R(\rho)\parallel_{tr}=1/d < 1$ as expected since $\rho$ is separable. Conversely, according to Eq. \eqref{eq-special-cases-of-R}, the maximally entangled state $\rho=\ketbra{\Omega}{\Omega}/d$ is mapped to $R(\rho)=\mathds{1} \otimes \mathds{1}/d$, so that  $\parallel~R(\rho)\parallel_{tr}=d>1$ and the entanglement of $\rho$ is well detected in this case.

Finally, it is worth adding that, by inspection, definition~\eqref{realignmentmap} of the realignment map can be decomposed as
\begin{equation}
R\big(A\otimes B\big)=\Big(\big(A\otimes B^T\big)\,F\Big)^{T_2}
\end{equation}
where $(\cdot)^{T_2}$ denotes a partial transposition on the second subsystem ($B$), and
$F = \sum_{i,j} \ket{ij}\bra{ji} =\ketbra{\Omega}{\Omega}^{T_{2}}$ is the exchange operator \cite{WolfThesis}. 
 From this, we obtain the following.
\begin{rem}\label{RandF}
	For any state $\rho$, the realignment map can be defined as 
	\begin{equation}
	\label{ReF}
	R(\rho) = \left(  \rho^{T_2} F \right)^{T_2}=(\rho F)^{T_2}F.
	\end{equation}
\end{rem}
In other words, the map $R$ boils down to the concatenation of partial transposition on subsystem $B$, then applying the exchange operator $F$ to the right, followed by partial transposition on subsystem $B$ again. Conversely, the roles of $F$ and $(\cdot)^{T_2}$ can be exchanged. This alternative definition of $R$ allows us to express the trace norm as
\begin{equation}
\label{RePPT}
\parallel R(\rho)\parallel_{tr}
=\parallel(\rho F)^{T_2}F\parallel_{tr}=\parallel(\rho F)^{T_2}\parallel_{tr},
\end{equation}
where the last equality comes from the fact that, for any operator $A$, we have 
\begin{eqnarray}
\parallel A F\parallel_{tr}&=&\tr\sqrt{AF(AF)^\dag}=\tr\sqrt{AFF^\dag A^\dag}\nonumber\\
&=&\tr\sqrt{AA^\dag}=\parallel A\parallel_{tr}
\label{eq-identity-F}
\end{eqnarray}
since  $FF^\dag=FF=\mathds{1}$.  From Eq. \eqref{RePPT}, it becomes obvious that for the special case of states $\rho_{s}$ belonging to the symmetric subspace, i.e.,  states satisfying $F\rho_{s}=\rho_{s} F=\rho_{s}$, the realignment criterion coincides with the PPT criterion \cite{toth}. Indeed,
$
\parallel R(\rho_{s})\parallel_{tr} = \parallel\rho_{s}^{T_2}\parallel_{tr}
$
and $\parallel\rho_{s}^{T_2}\parallel_{tr}=\sum_i |\lambda'_i |>1$ implies that at least one eigenvalue $\lambda'_i$ of the partial-transposed state $\rho_{s}^{T_2}$ is negative, since $\tr (\rho_{s})=\tr (\rho_{s}^{T_2})=\sum_i \lambda'_i =1$ (which is the PPT criterion). Beyond the case of states in the symmetric subspace, however, the realignment and PPT criteria are generally incomparable criteria.

Of course, the dual realignment map $R^{\dagger}$ can also be defined similarly as in Eq.~\eqref{ReF}, namely
	\begin{equation}\label{realignmentdualSWAP}
R^\dag(\rho) =  \left( F \rho^{T_2} \right)^{T_2} = F \left( F \rho \right)^{T_2}.
	\end{equation}
The difference with the (primal) realignment map $R$ is that the exchange operator $F$ is applied to the left. To be complete, let us mention that maps $R$ and  $R^{\dagger}$ can also be defined using partial transposition on the first subsystem denoted as $(\cdot)^{T_1}$, namely,
	\begin{eqnarray}\label{realignment-map-T1}
	R(\rho) &=& (F\rho^{T_1})^{T_1}=F(F\rho)^{T_1} ,  \nonumber \\
R^\dag(\rho) &=& \left(  \rho^{T_1} F \right)^{T_1} = \left( \rho F \right)^{T_1} F.
	\end{eqnarray}

\section{Weak realignment criterion}
Let us now introduce the weak realignment criterion, which is in general not as strong as the original realignment criterion but has the advantage of being easily computable and physically implementable using standard optical components. The weak realignment criterion applies to all states but our main focus in this paper will be its application to continuous-variable states, see Sec. \ref{sect-realignement-continuous-variable}. We will in particular provide some explicit calculations in the case of  $n\times n$ mode Gaussian states, in which case it boils down to computing a simple quantity that only depends on the covariance matrix of the state. As expected, however, the easiness of computation comes with the price of a lower entanglement detection sensitivity than the one of the original realignment criterion for Gaussian states as calculated in \cite{Zhang}. As a way to overcome this problem, we show below that for a special type of states that we call \textit{Schmidt-symmetric}, the weak form and original form of the realignment criterion become equivalent. This suggests the use of a filtration procedure for augmenting the detection sensitivity. As explored in Sec. \ref{sectionstep1}, we may ``symmetrize'' the state by locally applying a noiseless amplifier or attenuator (it does not affect the separability of the state, so we may apply the weak realignment criterion on the filtered state).

\subsection{Formulation}

 It is well known that the trace norm of an operator is greater than or equal to its trace (and we have equality if and only if the operator is positive semidefinite). Using Eq. \eqref{ReF},  we have that for any state
\begin{equation}
\begin{aligned}
\parallel R(\rho) \parallel_{tr} & \geq \tr \, R(\rho)  \\
&  = \tr  \left(  \rho^{T_2} F\right) \\
&  = \tr \left( \rho \, F^{T_2}  \right) \\
&  = \tr \left( \rho\, \ketbra{\Omega}{\Omega}  \right) =\bra{\Omega}\rho\ket{\Omega}
\end{aligned}
\end{equation}
where we have used the invariance of the trace under partial transposition $(\cdot)^{T_2}$ (line 2), the identity $\tr(A \, B^{T_2})=\tr(A^{T_2} B)$ for any bipartite operators $A$ and $B$ (line 3), and the definition of $F$ (line~4). 
Note that this result can also  be obtained by noticing that 
\begin{equation}
\tr(R(\rho)\, \mathds{1}\otimes \mathds{1})=\tr(\rho \, R^\dag(\mathds{1}\otimes \mathds{1}))=\tr(\rho \, \ketbra{\Omega}{\Omega}).
\end{equation}
We can thus state the following theorem:
\begin{theorem}[\textbf{Weak realignment criterion}]
	\label{weak1}
	For any bipartite state $\rho$, the trace norm of the realigned state can be lower bounded as
	\begin{equation}
	\parallel R(\rho) \parallel_{tr}\geq \tr \, R(\rho)=\mean{\Omega|\rho|\Omega}  .
	\label{eq-max-entanglement-component}
	\end{equation}
	Hence, if  $\rho$ is separable, then $\mean{\Omega|\rho|\Omega} \le 1$. Conversely,	
	if $\mean{\Omega|\rho|\Omega}> 1$, then $\rho$ is entangled. 
\end{theorem}

In other words, the weak realignment criterion amounts to computing the fidelity of state $\rho$ with respect to $\ket{\Omega}$.
It is immediate that its entanglement detection capability can only be lower than that of the original realignment criterion, Theorem~\ref{theoreal2}.
Furthermore, if we deal with bound-entangled states, the weak realignment criterion cannot detect entanglement. Indeed, we can link the weak realignment criterion with the PPT criterion by expressing
	\begin{equation}\label{reT2PPT}
	\parallel R(\rho)^{T_2}\parallel_{tr} 
	= \parallel \rho^{T_2} F\parallel_{tr} = \parallel \rho^{T_2}\parallel_{tr}  ,
	\end{equation}
	where we have used Eqs. \eqref{ReF} and \eqref{eq-identity-F}, combined with the inequality
		\begin{equation}
	 \parallel  R(\rho)^{T_2}\parallel_{tr} 
	\geq \tr \left( R(\rho)^{T_2} \right) = \tr \, R(\rho) , 
	\end{equation}
	Thus,
		\begin{equation}
	\parallel  \rho^{T_2}\parallel_{tr} \, 
	\geq \tr \, R(\rho) , 
	\end{equation}
	and we deduce that the weak realignment criterion is weaker than the PPT criterion. If a state is bound entangled, we have $	\parallel  \rho^{T_2}\parallel_{tr} =1$ which then implies that $\tr \, R(\rho)\leq 1$, so its entanglement cannot  be detected with the weak realignment criterion. 


It is instructive to apply the weak realignment criterion on each component of the operator Schmidt decomposition of $\rho$. Using Eq. \eqref{RmapwithEPR}, we have
\begin{equation}
\tr \, R(A \otimes B) = \bra{\Omega} A \otimes B \ket{\Omega} = \tr (A B^T) , 
\end{equation}
which implies that if $\rho =  \sum_i \lambda_i \, A_i \otimes B_i$, then
\begin{equation}
\tr \, R(\rho) = \sum_i \lambda_i \, \tr (A_i B_i^T) = \sum_i \lambda_i \, \langle B_i^* |A_i  \rangle
\end{equation}
Remembering that $\parallel R(\rho) \parallel_{tr} = \sum_i \lambda_i $, it appears that we must have $B_i=A_i^*$ in order to reach a situation where $\tr \, R(\rho) =\parallel R(\rho) \parallel_{tr} $. This is analyzed now.

\subsection{Schmidt-symmetric states}
\label{subsec-Schmidt-symmetric}

We  now show that for \textit{Schmidt-symmetric} states, the weak and original forms of the realignment criterion become equivalent (while the weak form is much simpler to compute). Let us define Schmidt-symmetric states $\rho_{sch}$ as the states that admit an operator Schmidt decomposition with $B_i=A_i^*$, $\forall i$, namely,
\begin{equation}
\rho_{sch}=\sum_{i}\lambda_i  \, A_i\otimes A_i^*.
\end{equation}
These states satisfy $F\rho_{sch}\, F=\rho_{sch}^*$ since applying $F$ on both sides is equivalent to exchanging the two subsystems and since the Schmidt coefficients are real. Note that the converse is not true as there exist states $\rho$ that satisfy $F\rho F=\rho^*$ but are not Schmidt-symmetric, for example the state $\rho = \sum_{i}\lambda_i  \, A_i\otimes (- A_i^*)$.
For any state $\rho$ that satisfies $F\rho F=\rho^*$, it is easy to see that  $R(\rho)$ is Hermitian since\footnote{Be aware that $R(\rho)^\dag$ is distinct from the dual map $R^\dag(\rho)$.}
\begin{eqnarray}
R(\rho)^\dag&=&((\rho F)^{T_2}F)^\dag \nonumber \\
&=&F(\rho^* F)^{T_1}  \nonumber \\
&=&F(F\rho)^{T_1}=R(\rho)  ,
\end{eqnarray}
where we have used Eqs. \eqref{ReF} and \eqref{realignment-map-T1}.
Thus, $R(\rho_{sch})$ is necessarily an Hermitian operator.

Actually, using the definition (\ref{realignmentmap2}) of the realignment map $R$, it appears that $R(\rho_{sch})=\sum_i\lambda_i\ket{A_i}\bra{A_i}$ is  positive semidefinite, so that  $ \parallel R(\rho_{sch})\parallel_{tr}=\tr \, R(\rho_{sch})$.
Conversely, if the latter equality is satisfied for a state $\rho$, it means that $R(\rho)$ is  positive semidefinite so it can be written as $R(\rho)=\sum_i\lambda_i\ket{A_i}\bra{A_i}$, which is nothing else but the realignment of a Schmidt-symmetric state. We have thus proven the following theorem:
\begin{theorem}[\textbf{Schmidt-symmetric states}] \label{schmidt}
A bipartite state $\rho$ is Schmidt-symmetric (i.e., it admits the operator Schmidt decomposition $\rho = \sum_{i}\lambda_i  \, A_i\otimes A_i^*$)  if and only if
\begin{equation}
 \parallel R(\rho_{})\parallel_{tr}=\tr \, R(\rho_{}).
 \end{equation}
\end{theorem}
This entails  the coincidence between the weak form of the realignment criterion derived in Theorem \ref{weak1} and the original realignment criterion of Theorem \ref{theoreal2} in the special case of Schmidt-symmetric states.

Incidentally, we note that the (necessary) condition $F\rho F=\rho^*$ for a state to be Schmidt-symmetric resembles the (necessary and sufficient) condition $F\rho F=\rho$ for a state to be symmetric under the exchange of the two systems. 
For this reason, when building a filtration procedure in order to bring the initial state closer to a Schmidt-symmetric state (see Sec. \ref{sectionstep1}), we will ``symmetrize" the state.
More precisely, we will exploit the fact that the condition $F\rho F=\rho^*$ implies that $\tr_1 \, \rho = \tr_2 \, \rho^*$. In other words, Schmidt-symmetric states are such that the reduced states of both subsystems are complex conjugate of each other, namely 
$\rho_{sch,2}=\rho_{sch,1}^*$, which is also a simple consequence of
\begin{eqnarray}
\rho_{sch,1}&=&\tr_2(\rho_{sch})=\sum_i\lambda_i \, A_i \, \tr A_i^*,  \nonumber\\
\rho_{sch,2}&=&\tr_1(\rho_{sch})=\sum_i\lambda_i \, A_i^* \, \tr A_i
\end{eqnarray}
Hence, they have the same eigenspectrum since their eigenvalues are real, and in particular the same purity (but the converse is not true),
\begin{eqnarray}
\tr(\rho_{sch,1}^2) = \tr(\rho_{sch,2}^2) .
\label{eq-same-purity}
\end{eqnarray}
The filtration procedure that we  apply in Sec. \ref{sectionstep1} follows Eq. \eqref{eq-same-purity} in the sense that we will ``symmetrize'' the initial state so that the two subsystems reach the same purity.



\section{Weak realignment criterion for continuous-variable states}
\label{sect-realignement-continuous-variable}

\subsection{Preliminaries and symplectic formalism}

Now, we turn to the application of the weak realignment criterion for continuous-variable states (i.e., living in an infinite-dimensional Fock space).
We start by briefly introducing  the symplectic formalism employed for continuous-variable states. More details can be found, for example in \cite{weedbrook, Anaellethesis}.

A continuous-variable system is represented by $N$ modes, each of them associated with a Hilbert space spanned by the Fock basis and having its own mode operators $a_i$ and $a_i^\dag$ which verify the commutation relation $[a_i,a_i^\dag]=1$. We define  the quadratures vector $\mathbf{r}=(x_1,p_1,x_2,p_2,\cdots,x_N,p_N)$ where
\begin{equation}
x_i=\frac1{\sqrt{2}}(a_i+a_i^\dag),\quad p_i=-\frac1{\sqrt{2}}(a_i-a_i^\dag)\quad \forall i=1,\cdots,N.
\end{equation}

Each quantum state $\rho$ can be described by a quasiprobability distribution function, the Wigner function
\begin{equation}
W(\mathbf{x},\mathbf{p})=\frac{1}{(2\pi)^N}\int d\mathbf{y}e^{-i\mathbf{p}\cdot \mathbf{y}}\langle \mathbf{x}+\mathbf{y}/2|\rho|\mathbf{x}-\mathbf{y}/2\rangle
\end{equation}
which is normalized to one.

The first-order moments constitute the displacement vector, defined as $\mean{\mathbf{r}}=\tr (\mathbf{r}\rho)$, while the second moments make up the covariance matrix $\gamma$ whose elements are given by
\begin{equation}
\gamma_{ij}=\frac12\mean{\{r_i,r_j\}}-\mean{r_i}\mean{r_j}.
\end{equation}
where $\{\cdot,\cdot\}$ represents the anticommutator.

	A \textit{Gaussian} state is fully characterized by its displacement vector and covariance matrix and its Wigner function has a Gaussian shape. Some relevant examples of Gaussian states are the following:
	\begin{itemize}
		\item The coherent state $\ket{\alpha}$  is a displaced vacuum state (where $\alpha=0$), meaning that the covariance matrix is the one of the vacuum $\gamma_{\ket{\alpha}}=\gamma_{\ket{0}}=\frac12\begin{psmallmatrix}1&0\\0&1
		\end{psmallmatrix}$ but the first moment depends on the value of $\alpha$.
		\item The squeezed state $\ket{r}$: 
	the uncertainty of one quadrature is minimized by squeezing it according to the squeezing parameter $r$; the covariance matrix is given by $\gamma_{\ket{r}}=\frac12\begin{psmallmatrix}e^{-2r}&0\\0&e^{2r}
		\end{psmallmatrix}$.
		\item The thermal state $\rho_{th}$  is a mixed state where the uncertainties of each quadratures are equal, but not minimals; the covariance matrix is given by $\gamma_{th}=\frac12\begin{psmallmatrix}2\mean{n}+1&0\\0&2\mean{n}+1
		\end{psmallmatrix}$ where $\mean{n}$ is the mean photon number.
		\item The two-mode squeezed vacuum state $\ket{TMSV}$ is a two-mode state with covariance matrix
		\begin{equation}
\gamma_{TMSV}=\frac12\begin{pmatrix}\cosh 2r &0&\sinh 2r&0\\0&\cosh 2r&0&-\sinh 2r\\\sinh 2r &0&\cosh 2r&0\\0&-\sinh 2r&0&\cosh 2r
\end{pmatrix}. \label{gammaTMSV}
		\end{equation}
		If the squeezing $r$ tends to infinity, one recovers the well-known Einstein-Podolsky-Rosen (EPR) state.
	\end{itemize}

A Gaussian unitary transformation is a unitary transformation that preserves the Gaussian
character of a quantum state. In terms of quadrature
operators, a Gaussian unitary transformation is described by the map 
\begin{equation}
\mathbf{r}\to \mathcal{S}\mathbf{r}+\mathbf{d},
\end{equation}
where $\mathbf{d}$ is a real vector of dimension $2N$ and $\mathcal{S}$ is a real $2N\times2N$ matrix which is symplectic.

\subsection{Examples of realigned states}

 It is instructive first to check the action of the realignment map $R$ on some of the well-known states of quantum optics:
\begin{itemize}
	\item Fock states\footnote{Remember that the Fock basis $\{\ket{n}\}$ is used as the preferred basis with respect to which the realignment map $R$ is defined.}:\\ $R(\ket{n_1}\bra{n_2}\otimes \ketbra{n_3}{n_4})=\ketbra{n_1}{n_3}\otimes\ketbra{n_2}{n_4}$.
	\item Position states\footnote{This can be proven using Eq.~(\ref{RmapwithEPR}) and expressing $\ket{\Omega}=\sum_n\ket{n,n}$ in the position basis, namely $\ket{\Omega}=\int dx\, \ket{x,x}$}: \\ $R(\ket{x_1}\bra{x_2}\otimes \ketbra{x_3}{x_4})=\ketbra{x_1}{x_3}\otimes\ketbra{x_2}{x_4}$.
	\item Momentum states\footnote{This can be proven using Eq.~(\ref{RmapwithEPR}) and expressing $\ket{\Omega}=\sum_n\ket{n,n}$ in the momentum basis: $\ket{\Omega}=\int dp\, \ket{p,-p}$}:\\  $R(\ket{p_1}\bra{p_2}\otimes \ketbra{p_3}{p_4})=\ketbra{p_1}{-p_3}\otimes\ketbra{-p_2}{p_4}$.
	\item Coherent states:\\ $R(\ket{\alpha}\bra{\beta}\otimes \ketbra{\gamma}{\delta})=\ketbra{\alpha}{\gamma^*}\otimes\ketbra{\beta^*}{\delta}$. In particular, $R(\ket{\alpha}\bra{\alpha}\otimes \ketbra{\alpha^*}{\alpha^*})=\ketbra{\alpha}{\alpha}\otimes\ketbra{\alpha^*}{\alpha^*}$, so that a pair of phase-conjugate coherent states is invariant under $R$.
	\item Two-mode squeezed vacuum state:\\
	Defining $\ket{TMSV}=~(1-\tau^2)^{1/2} \sum_{i}\, \tau^{i} \, \ket{i}\ket{i}$ with $0\le \tau <1$ characterizing the squeezing, we obtain  $R(\ketbra{TMSV}{TMSV})=\frac{1+\tau}{1-\tau} \, \rho_{th}\otimes\rho_{th}$ where   $\rho_{th}=(1-\tau) \sum_i \, \tau^i\ketbra{i}{i}$ is a thermal state. Entanglement is detected in this case since $\parallel~R(\ket{TMSV}\bra{TMSV})\parallel_{tr}=\frac{1+\tau}{1-\tau}>1$ as soon as $\tau>0$.
	\item Tensor product of thermal states:\\ $R(\rho_{th}\otimes\rho_{th})=\frac{1-\tau}{1+\tau}\, \ketbra{TMSV}{TMSV}$ so that we have $\parallel R(\rho_{th}\otimes\rho_{th})\parallel_{tr}=\frac{1-\tau}{1+\tau} \le 1$, as expected for a separable state.
\end{itemize}

\subsection{Expression of $\tr (R)$ for arbitrary states}

Let us now show how the weak form of the realignment criterion provides us with an implementable entanglement witness for all continuous-variable states. According to Eq. \eqref{eq-max-entanglement-component}, in order to access $\tr\,R(\rho)$ we need to project state $\rho$ onto $\ket{\Omega}$, which can be thought of as an unnormalized infinitely entangled two-mode vacuum squeezed state. As proven in Appendix \ref{AppendixA}, the latter can be reexpressed as 
\begin{equation}
 \ket{\Omega}=\sqrt{\pi} \, U_{BS}^\dagger \ket{0}_{x_1}\ket{0}_{p_2} ,
\end{equation}
that is, it can formally be obtained by applying (the reverse of) a 50:50 beam splitter Gaussian unitary $U_{BS}$ on an input state of the product form $\ket{0}_{x_1}\ket{0}_{p_1}$, where $U_{BS}\ket{z}_{x_1} \ket{z'}_{x_2} =\left|(z-z')/\sqrt{2}\right\rangle_{x_1}  \left|(z+z')/\sqrt{2}\right\rangle_{x_2}$ in the position eigenbasis and $\ket{0}_{x_1}$ (resp. $\ket{0}_{p_2}$) is the position (momentum) eigenstate with zero eigenvalue. Therefore,
\begin{align}
\tr \, R(\rho) = \pi \, \bra{0}_{x_1} & \bra{0}_{p_2} \, U_{BS} \, \rho \, U_{BS}^\dag \, \ket{0}_{x_1}\ket{0}_{p_2}.
\label{moyenneEPR2modes}
\end{align}
Hence, implementing the weak realignment criterion amounts to expressing the probability of projecting the state $\rho'=U_{BS}\, \rho \, U_{BS}^\dag$ onto $\ket{0}_{x_1}\ket{0}_{p_2} $ where $\rho'$ is the state obtained at the output of a 50:50 beam splitter (see Fig. \ref{prob} for the two-mode case). This yields an experimental way of constructing an entanglement witness using standard optical components since entanglement is detected simply by applying a Gaussian measurement on the state \cite{weedbrook,fabre}.

\begin{figure}
	\includegraphics[trim= 5.5cm 5cm 5cm 19cm ,clip,width=0.4\textwidth]{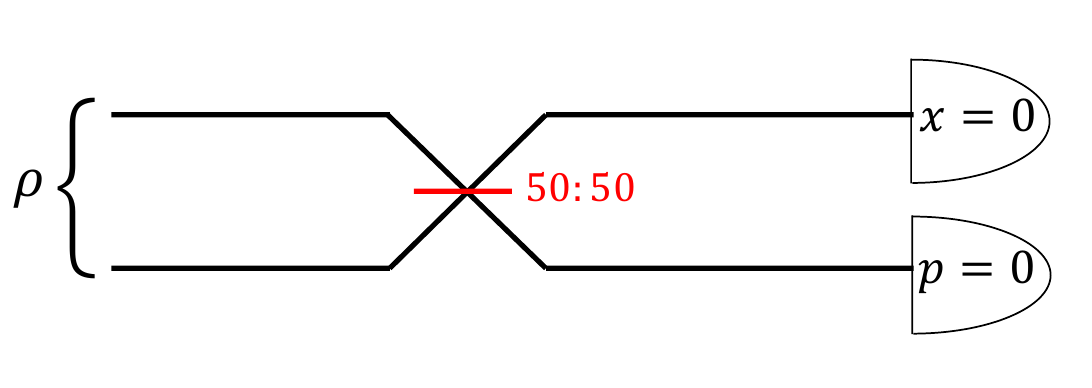} \caption{\label{prob} Weak realignment criterion for a (two-mode) state $\rho$. The trace of the realigned state $R(\rho)$ is obtained by computing the probability of measuring $x_1=p_2=0$ on the output state after processing $\rho$ through a $50:50$ beam splitter, see Eq. \eqref{moyenneEPR2modes}.}
\end{figure}

Furthermore, this entanglement witness can be generalized to $n\times n$ modes with quadrature components $\mathbf{x}_A=(x_1,\cdots,x_n)$, $\mathbf{p}_B=(p_{n+1},\cdots,p_{2n})$. We have\footnote{To be more precise, if the $n$ first modes belong to Alice and the $n$ last modes belong to Bob, we apply $n$ 50:50 beam splitters between Alice's $i$th mode and Bob's $i$th mode, for $i=1,...,n$. }  
 \begin{equation}
 \ket{\Omega_{n\times n}}=\pi^{n/2}\, U_{BS}^\dag \, \ket{ 0}_{\mathbf{x}_A} \ket{0 }_{\mathbf{p}_B}   ,
 \end{equation}
 with the short-hand notation $\ket{ 0}_{\mathbf{x}_A} \equiv \ket{0, \cdots, 0}_{\mathbf{x}_A}$ and $\ket{0 }_{\mathbf{p}_B} \equiv \ket{0, \cdots, 0}_{\mathbf{p}_B}$,
hence
\begin{equation}
\tr \, R(\rho) =
\pi^n \, \bra{0}_{\mathbf{x}_A}  \bra{0}_{\mathbf{p}_B} \, \rho' \, \ket{ 0}_{\mathbf{x}_A}\ket{0 }_{\mathbf{p}_B} .
\label{moyenneEPR}
\end{equation}



\subsection{Expression of $\tr(R)$ for Gaussian states}
\label{sec:gaussian}

If the initial state $\rho$ is  an $n\times n$ Gaussian state, the state $\rho' = U_{BS}\, \rho \, U_{BS}^\dag$ will be Gaussian too (since the beam splitter is a Gaussian unitary). Its Wigner function is thus given by
\begin{equation}
W_{\rho'}(\mathbf{r})=\frac{1}{(2\pi)^{2n}\sqrt{\det\gamma'}}e^{-\frac{1}{2}\mathbf{r}(\gamma')^{-1}\mathbf{r}^T}
\end{equation}
where $\mathbf{r}=(x_1,p_1,x_2,p_2,\cdots,x_{2n},p_{2n})$ and $\gamma'$ is the covariance matrix of $\rho'$ obtained as
\begin{equation}
	\gamma'=\mathcal{S}\gamma\mathcal{S^T}\qquad\text{with}\qquad\mathcal{S}=\frac{1}{\sqrt{2}}\left(
		\begin{array}{cccc}
	\mathds{1}_{2n} & -\mathds{1}_{2n} \\
		\mathds{1}_{2n} & \mathds{1}_{2n} \\
		\end{array}
		\right),
\end{equation}
being the symplectic matrix representing the beam splitting transformation and $\gamma$ being the covariance matrix of $\rho$. 
The probability of projecting $\rho'$ onto $\ket{ 0}_{\mathbf{x}_A}\ket{0 }_{\mathbf{p}_B}$ as of Eq.~(\ref{moyenneEPR}) is thus easy to compute.
Indeed,  the probability distribution of measuring $\mathbf{x}_A$ on the $n$ modes of the first system and $\mathbf{p}_B$ on the $n$ modes of the second system is given by\footnote{The probability distribution is Gaussian since we are dealing with Gaussian states.}
\begin{equation}
\label{probGauss}
P(\mathbf{x}_A,\mathbf{p}_B)
=\frac{1}{(2\pi)^n\sqrt{\det \gamma_{w}}} e^{-\frac{1}{2}\left(\begin{smallmatrix} \mathbf{x}_A,\mathbf{p}_B 	\end{smallmatrix}\right)\gamma_{w}^{-1}\left(\begin{smallmatrix}	 \mathbf{x}_A,\mathbf{p} _B	\end{smallmatrix}\right)^T}
\end{equation}
where $\gamma_{w}$ ("$w$" is for witness) is the restricted covariance matrix obtained by removing the lines and columns of the unmeasured quadratures of $\gamma'$ (see Fig.~\ref{matreduce} for examples with $n=1$ and $2$). Thus $\bra{\Omega}\rho\ket{\Omega}=\pi^n P(\mathbf{0,0})=\frac{1}{2^n\sqrt{\det\gamma_{w}}}$. In Appendix \ref{AppendixProb}, we show how $P(0,0)$ can also be computed directly in a two-mode case ($n=1$).

\begin{figure}
	\includegraphics[width=0.4\textwidth]{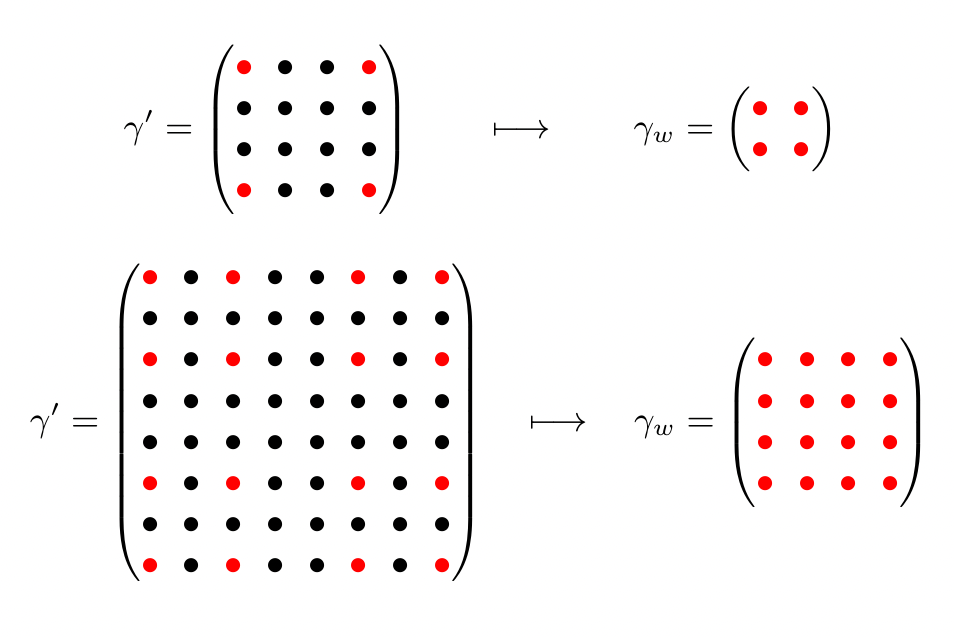}
	\caption{\label{matreduce}Construction of the restricted covariance matrix $\gamma_{w}$ from the covariance matrix $\gamma'$, when $n=1$ and $n=2$. The red bullets correspond to the entries of $\gamma'$ that are copied in $\gamma_{w}$ while the black bullets are the entries  that are dropped.}
\end{figure}

We are now ready to state the following theorem:
\begin{theorem}[\textbf{Weak realignment criterion for Gaussian states}]
	\label{weak}
For any $n\times n$ Gaussian state $\rho^G$, the trace norm of the realigned state can be lower bounded as
\begin{equation}
\parallel R(\rho^G) \parallel_{tr}\geq\tr \, R(\rho^G)=\frac{1}{2^n\sqrt{\det\gamma_{w}}} \label{formuleRealCriterion}.
\end{equation}
Hence, if $\rho$ is separable, then $ \frac{1}{2^n\sqrt{\det\gamma_{w}}} \le 1$.  Conversely, 
\begin{equation}
\text{if } \frac{1}{2^n\sqrt{\det\gamma_{w}}} > 1\text{,   then } \rho^G \text{ is entangled.} 
\label{conditionweak}
\end{equation}
\end{theorem}

Incidentally, we note that condition (\ref{conditionweak}) is equivalent to $\det \gamma_{w} <1/4^n$ and can thus be viewed as checking the nonphysicality of  $\gamma_{w} $ via the violation of the Schrödinger-Robertson uncertainty relation. This is in some sense similar to the PPT criterion, which is based on checking the nonphysicality of the partially transposed state.  

Consider the special case of a two-mode Gaussian state ($n=2$). Its covariance matrix can always be transformed into the normal form \cite{duan} 
\begin{equation}
\gamma^G=~\left(
\begin{array}{cccc}
a & 0 & c & 0 \\
0 & a & 0 & d \\
c & 0 & b & 0 \\
0 & d & 0 & b \\
\end{array}
\right)
\label{normalform}
\end{equation}
by applying local Gaussian unitary operations\footnote{The covariance matrix of the TMSV state, Eq. \eqref{gammaTMSV} is an example of the normal form}, which are combinations of squeezing
transformations and rotations and hence do not influence the separability
of the state. Applying Theorem \ref{weak}, we get 
\begin{equation}
\tr \, R(\rho^G)=\frac{1}{\sqrt{(a+b-2 c) (a+b+2 d)}}
\label{eq-trace-realigned-normal}
\end{equation}
and the weak realignment criterion reads 
\begin{equation}
\frac{1}{\sqrt{(a+b-2 c) (a+b+2 d)}} > 1 \Rightarrow \rho^G \text{ is entangled.} \nonumber
\end{equation}
In comparison, it was shown in \cite{Zhang} that for a covariance matrix in the normal form, Eq.~(\ref{normalform}), the trace norm of the realigned state is given by
\begin{equation}
\parallel R(\rho^G) \parallel_{tr}=\frac{1}{2\sqrt{\left(\sqrt{ab}-|c|\right)\left(\sqrt{ab}-|d|\right)}}.
\label{formuleZhang}
\end{equation}
Comparing Eqs. \eqref{eq-trace-realigned-normal} and \eqref{formuleZhang} illustrates the fact that the weak realignment criterion is generally weaker than the original form of the realignment criterion (there exist states such that $\parallel R(\rho^G) \parallel_{tr} \, >1$ while $\tr \, R(\rho^G)\le 1$).

As already mentioned, both criteria become equivalent if the state is in a Schmidt-symmetric form. In this case, for a general $n\times n$ Gaussian state $\rho^G$ described by the covariance matrix 
\begin{equation}
	\gamma^G=\begin{pmatrix}
A&C\\C^T&B
	\end{pmatrix}  ,
	\label{covmatrix}
\end{equation} 
it implies that both reduced covariance matrices must be identical, namely, $A=B$. Indeed, $\rho^G$ being Schmidt-symmetric implies that 
$F\rho^G\, F=(\rho^G)^*$. Exchanging Alice and Bob's systems yields a Gaussian state $F \rho^G F$ of covariance matrix
$\begin{pmatrix}
B&C^T\\C&A
	\end{pmatrix} $,
while $(\rho^G)^*=(\rho^G)^T$ is a Gaussian state that admits the covariance matrix
$\begin{pmatrix}
A&C^T\\C&B
	\end{pmatrix}$.
Identifying these two covariance matrices, we conclude that any Schmidt-symmetric Gaussian state must have a covariance matrix of the form
\begin{equation}
	\gamma_{sch}^G=\begin{pmatrix}
A&C\\C^T&A
	\end{pmatrix} .
\label{covmatrix-Gauss-Schmidt-sym}
\end{equation} 
In particular, both reduced covariance matrices have the same determinant, i.e., $\det A=\det B$, which is expected since we know from Eq. \eqref{eq-same-purity} that the two reduced states have the same purity, $\tr((\rho^G_{1})^2) =~\frac{1}{2^n\sqrt{\det A}}$ and $\tr((\rho^G_{2})^2) =\frac{1}{2^n\sqrt{\det B}}$. 
In Sec. \ref{sectionstep1}, we  apply a filtration procedure that brings the Gaussian state closer to a Schmidt-symmetric Gaussian state, which will have the effect of bringing the covariance matrix  \eqref{covmatrix} closer to the form \eqref{covmatrix-Gauss-Schmidt-sym}. More precisely, we will consider a filtration that equalizes the determinants of the reduced covariance matrices (hence, the two subsystems reach the same purity). We  say that a covariance matrix of the form~\eqref{covmatrix} has been \textit{symmetrized} when $\det A=\det B$.

Note that the covariance matrix in form \eqref{covmatrix-Gauss-Schmidt-sym} is a necessary but not sufficient condition for a Gaussian state to be Schmidt-symmetric. A necessary and sufficient condition must imply additional constraints on matrix $C$. Let us show this for a two-mode Gaussian state with covariance matrix in the normal form
\begin{equation}
\gamma=\begin{pmatrix}
a&0&c&0\\0&a&0&d\\c&0&a&0\\0&d&0&a
\end{pmatrix},
\label{normalform+a=b}
\end{equation} 
which is a special case of Eq. \eqref{covmatrix-Gauss-Schmidt-sym}. Using Eq. \eqref{eq-trace-realigned-normal}, we obtain
\begin{equation}
\tr(R(\rho))=\frac{1}{2\sqrt{(a-c)(a+d)}},
\end{equation} 
while Eq. \eqref{formuleZhang} implies that
\begin{equation}
  \parallel R(\rho)\parallel_{tr}=\frac{1}{2\sqrt{(a-|c|)(a-|d|)}}. 
\end{equation}  
Both formulas are thus equivalent only if $c \ge 0$ and $d \le 0$, which gives the additional constraint on $C$. Thus, the necessary and sufficient condition for a two-mode state with covariance matrix in normal form \eqref{normalform} to be Schmidt-symmetric is that $a=b$,  $c \ge 0$, and $d \le 0$.

This last point can be illustrated by considering $\ket{\Omega}=\sum_n\ket{n}\ket{n}=\int dx\, \ket{x,x}=\int dp\, \ket{p,-p}$, which can be viewed (up to normalization) as the limit of a two-mode squeezed vacuum state with infinite squeezing. It has $c>0$ and $d<0$ since the $x$'s are correlated and $p$'s are anticorrelated. It admits an operator Schmidt decomposition
$\ket{\Omega}\bra{\Omega}= \sum_{n,m}\ket{n}\bra{m}\otimes\ket{n}\bra{m}$
with all Schmidt coefficients being equal to one and the associated operators  $ A_{n,m}=B_{n,m}= \ket{n}\bra{m}$; hence it is Schmidt-symmetric since it satisfies $B_{n,m}=A_{n,m}^*$. Now, let us apply a phase shift of $\pi$ on one of the modes, yielding $\ket{\Omega'}=\sum_n (-1)^n\ket{n}\ket{n}=\int dx\, \ket{x,-x}=\int dp\, \ket{p,p}$. Here, we have  $c<0$ and $d>0$ since the $x$'s are anticorrelated and $p$'s are correlated, so it should not be Schmidt-symmetric. Accordingly, it can be checked that $\ket{\Omega'}\bra{\Omega'}$ does not admit an operator Schmidt decomposition with  $B_{n,m}=A_{n,m}^*$. We may decompose it as $\ket{\Omega'}\bra{\Omega'}= \sum_{n,m} A_{n,m}\otimes B_{n,m}$
where all Schmidt coefficients are again equal to one and, for example, $A_{n,m}=B_{n,m}= i^{n+m} \ket{n}\bra{m}$ or $ A_{n,m}= (-1)^{n+m} B_{n,m}= \ket{n}\bra{m}$, but in all cases $B_{n,m} \ne A_{n,m}^*$. This is an example of an (unormalized) state verifying $F\rho F = \rho^*$ but that is not Schmidt-symmetric. Since $\ket{\Omega}$ and $\ket{\Omega'}$ share the same Schmidt coefficients, the trace norm of their realignments coincide and are equal to the trace of the realignment of $\ket{\Omega}$ only (in contrast, the trace of the realignment of $\ket{\Omega'}$ vanishes).

The link between $\ket{\Omega}$ and $\ket{\Omega'}$ suggests that a suitable local phase shift operation performed on one of the modes of a state can be useful to make the state closer to being Schmidt-symmetric, and hence to enhance the detection capability of the weak realignment criterion (an example of this feature is shown in Sec.~\ref{Section:Examples}).
Applying a local phase shift operation is, however, not always sufficient to make the state exactly Schmidt-symmetric, as can be seen by considering a covariance matrix of the form \eqref{normalform}, where we impose that $c\ge 0$ and $d \le 0$. Indeed, as soon as $a\ne b$, one can verify that $\tr(R(\rho)) < \, \parallel R(\rho)\parallel_{tr}$ as a consequence of the well-known inequality between arithmetic and geometric means, $\sqrt{ab} \le (a+b)/2$, which is saturated if and only if $a=b$.

\section{Improvement of the weak realignment criterion via filtration}
\label{sectionstep1}


According to Theorem~\ref{weak}, the trace norm of the realigned state $\parallel R(\rho)\parallel_{tr}$ is greater than (or equal to) $\tr \, R(\rho)$, which for $n\times n$ Gaussian states is a quantity that solely depends on the determinant of the restricted covariance matrix $\gamma_w$. For this reason, while it is easier to compute (especially in higher dimension), the weak realignment criterion has generally a lower entanglement detection performance than the realignment criterion (as applied in \cite{Zhang}). This suggests the possibility of improving the criterion by transforming the state via a suitable (invertible) operation prior to applying the criterion.

Since the trace norm and trace of the realigned state  are equivalent for a Schmidt-symmetric state, the natural idea is to find a procedure that ideally transforms the initial state into a Schmidt-symmetric state without of course creating or destroying entanglement. We focus here on $n\times n$ Gaussian states and exploit the fact that any Schmidt-symmetric Gaussian state admits a covariance matrix of the form \eqref{covmatrix-Gauss-Schmidt-sym}, in particular its reduced determinants are equal. Even if this is not a sufficient condition for a state to be Schmidt-symmetric, we choose to symmetrize the initial state by equalizing the reduced determinants of its covariance matrix in order to reach a state that is closer to (ideally equal to) a Schmidt-symmetric state. We then apply Theorem~\ref{weak} on the resulting symmetrized state in order to get an enhanced entanglement detection performance.

An  $n \times n$ Gaussian state $\rho$ is fully characterized by its displacement vector $\mathbf{d}$ and covariance matrix $\gamma$ defined in Eq.~\eqref{covmatrix}.
Since first-order moments are irrelevant as far as entanglement detection is concerned, we can restrict to states with $\mathbf{d}=0$ with no loss of generality. To symmetrize the state, we will exploit a filtering operation in the Fock basis as follows. Suppose that the first subsystem has a smaller noise variance or more precisely that 
 $\det A < \det B$ in Eq.~\eqref{covmatrix}, meaning that the purity of the first subsystem is larger than that of the second subsystem (the opposite case is treated below). We process each mode of the first subsystem through a (trace-decreasing) noiseless amplification map \cite{HNLA,universal-squeezer,He}, that is
\begin{equation}
\rho_{AB}\rightarrow \tilde\rho_{AB} = c \, (t^{\hat n/2}\otimes \mathds{1}) \rho_{AB}  (t^{\hat n/2}\otimes \mathds{1}),
\label{eq-hnla}
\end{equation}
where $c$ is a constant, $\hat n$ is the total photon number in the modes of the first subsystem, and $t>1$ is the transmittance or gain ($\sqrt{t}$ is the corresponding amplitude gain). It can be checked that this map effects an increase of the noise variance of the first subsystem (it increases $\det A$). Note that if the input state $\rho_{AB}$ is Gaussian, then the output state $\tilde\rho_{AB}$ remains Gaussian \cite{gagatsos}. Crucially, this map does not change the separability of the state (the amount of entanglement might change, but no entanglement can be created from scratch or fully destroyed). Therefore, $\tilde\rho_{AB}$ should be closer to a Schmidt-symmetric state and is a good candidate for applying Theorem~\ref{weak}.


To find the covariance matrix of the output state $\tilde\rho_{AB}$, we follow the evolution of the Husimi function defined as 
\begin{equation}
Q(\boldsymbol{\alpha})=\frac{1}{\pi^n}\bra{\boldsymbol{\alpha}}\rho\ket{\boldsymbol{\alpha}}
\end{equation} 
where $\ket{\boldsymbol{\alpha}}$ is a vector of coherent states. For an $n\times n$ Gaussian state $\rho_{AB}$, the Husimi function is given by
\begin{equation}
Q(\boldsymbol{\alpha,\beta})=\frac{1}{\pi^{2n}\sqrt{\det(\gamma+\frac{\mathds{1}}{2})}}e^{-\frac{1}{2} \boldsymbol{r}^T\Gamma \boldsymbol{r}}
\label{HusimiGaussian}
\end{equation}
where $\alpha$ is associated to the first system and $\beta$ to the second, $\Gamma=(\gamma+\mathds{1}/2)^{-1}$, and 
\begin{eqnarray}
\lefteqn{   \boldsymbol{r}=\sqrt{2}\Big(\Re(\alpha_1),\Im(\alpha_1),\cdots,\Re(\alpha_n),\Im(\alpha_n),  }  \hspace{2cm}  \nonumber \\
&& \Re(\beta_1), \Im(\beta_1),\cdots,\Re(\beta_n),\Im(\beta_n)\Big)^T 
\end{eqnarray}
with $\Re(\cdot)$ and $\Im(\cdot)$ representing the real and imaginary parts. The noiseless amplification map enhances the amplitude of a coherent state as $\ket{\alpha}\rightarrow e^{(t-1)|\alpha|^2 /2}\ket{\sqrt{t}\,\alpha}$. Therefore, the Husimi function of the output state $\tilde\rho_{AB}$ is equal to (see \cite{fiurasek2} for more details)
\begin{eqnarray}
\tilde Q(\boldsymbol{\alpha,\beta})&\propto&\frac{1}{\pi^{2n}}\bra{\boldsymbol{\alpha,\beta}}(t^{\hat n /2}\otimes \mathds{1})\rho_{AB}(t^{\hat n /2}\otimes \mathds{1})\ket{\boldsymbol{\alpha,\beta}}\nonumber\\
&=&e^{(t-1)(|\alpha_1|^2+\cdots+|\alpha_n|^2)}Q(\sqrt{t} \, \boldsymbol{\alpha,\beta}).
\end{eqnarray}
Since  the output state $\tilde\rho_{AB}$ is a Gaussian state, its Husimi function is still of the form~(\ref{HusimiGaussian}) with an output covariance matrix $\tilde{\gamma}$ (and corresponding $\tilde{\Gamma}$). Comparing the exponent of both expressions, we find that
\begin{eqnarray}
\tilde{\Gamma}&=&M\Gamma M-(M^2-\mathds{1})\\
\tilde{\gamma}&=&\left[M\left(\gamma+\frac{\mathds{1}}{2}\right)^{-1}M-(M^2-\mathds{1})\right]^{-1}-\frac{\mathds{1}}{2}\nonumber
\end{eqnarray}
where
\begin{equation}
M=\begin{pmatrix}
\sqrt{t}\,\mathds{1}_{2n\times 2n}&0\\0&\mathds{1}_{2n\times 2n}
\end{pmatrix}.
\end{equation}
The last point before applying the weak realignment criterion on $\tilde\rho_{AB}$ is to find a suitable value for the transmittance $t$ (note that $t$ must be greater than $1$). A simple ansatz is to choose $t$ so that the filtered state $\tilde\rho_{AB}$ is a symmetrized Gaussian state,  that is, the noise variance of both subsystems are equal ($\det  A= \det  B$).

Now, if the first subsystem has a larger noise variance (namely $\det A > \det B$), we can simply exchange the roles of $A$ and $B$ and apply the noiseless amplification map on the modes of the second subsystem. Alternatively, we may consider another filtering operation in the Fock basis by processing each mode of the first subsystem through a (trace-decreasing) noiseless attenuation map \cite{HNLAtt,gagatsos}. Formally, it is defined exactly as the noiseless amplification map in Eq. \eqref{eq-hnla} but with a transmittance $t<1$, so it leads to very similar calculations. Physically, the noiseless attenuation map has the advantage to admit an exact physical implementation (unlike the noiseless amplification map), which provides us with another method to compute the output covariance matrix $\tilde{\gamma}$ (and corresponding $\tilde{\Gamma}$).  Indeed, processing the state of a mode through a noiseless attenuation map is equivalent to processing it through a beam splitter of transmittance $t$ (with vacuum on an ancillary mode) and then postselecting the output conditionally on the vacuum on the ancillary mode (see Fig.~\ref{symm}). We give the details of this alternative calculation for the two-mode case in Appendix \ref{AppendixB}.

\begin{figure}
	\includegraphics[trim= 4cm 11cm 6cm 12cm ,clip,width=0.4\textwidth]{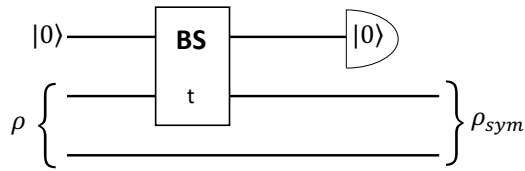} \caption{\label{symm} Circuit implementing the filtration on a two-mode Gaussian state $\rho$ : a noiseless attenuation map is applied on the first mode of $\rho$ (BS represents a beam splitter of transmittance $t$). The resulting state $\rho_{sym}$ is symmetrized if the value of $t$ is well chosen.}
\end{figure}

We note that the enhancement of the weak realignment criterion obtained via prior filtration can be viewed as the consequence of using $\tr R(\rho)= \tr(  \rho \, \ket{\Omega}\bra{\Omega} ) $ but with a better witness operator than $\ket{\Omega}\bra{\Omega}$. Let us define the filtration map as $\rho \to \Lambda_F (\rho)$. For example, consider the noiseless attenuation map $\Lambda_F (\rho) \propto (t^{\hat n/2}\otimes \mathds{1}) \rho  (t^{\hat n/2}\otimes \mathds{1})$ applied on a $1\times 1$ state (with $t<1$). The trace of the realigned state after filtration can be expressed as
\begin{eqnarray}
\tr( R(\Lambda_F (\rho)))&=& \tr(  \Lambda_F (\rho) \, \ket{\Omega}\bra{\Omega} ) \nonumber \\
&=& \tr(\rho \, \Lambda_F^\dag (\ket{\Omega}\bra{\Omega})) 
\end{eqnarray}
where $ \Lambda_F^\dag$ stands for the dual filtration map. In this example,  we note that $\Lambda_F^\dag  = \Lambda_F$ and  $\Lambda_F (\ket{\Omega}\bra{\Omega}))$ is proportional to the projector onto a two-mode squeezed vacuum state. In other words, the enhancement in this example is obtained by computing the fidelity of $\rho$ with respect to $\sum_n t^{n/2} \ket{n}\ket{n}$ instead of $\ket{\Omega}=\sum_n \ket{n}\ket{n}$.

In the next section, we  apply this filtration procedure on several examples of Gaussian states in order to show how the weak realignment criterion assisted with filtration can indeed improve entanglement detection.

\section{Applications}\label{Section:Examples}

\subsection{Two-mode squeezed vacuum state with Gaussian additive noise}
\label{ExampleEPR}
We first illustrate how the filtration procedure enables a better entanglement detection on two-mode entangled Gaussian states. In particular, we  show that for specific examples, computing the trace of the realigned state (after filtration) is equivalent to computing its trace norm. Let us consider a two-mode squeezed vacuum state whose first mode is processed through a Gaussian additive-noise channel as shown in Fig. \ref{exemple}. We denote $V$ the variance of this added noise. It is known that the entanglement  of the two-mode squeezed state decreases when we increase the noise variance, until $V=1$ at which point it becomes separable (if $V\geq1$, the channel is entanglement breaking \cite{holevo}).  Our Gaussian state has a covariance matrix
\begin{equation} \label{covEprplusbruit}
\gamma=\left(
\begin{array}{cccc}
V +\frac{\cosh 2r}{2} & 0 & \frac{\sinh 2r}{2} & 0 \\
0 & V +\frac{\cosh 2r}{2} & 0 & -\frac{\sinh 2r}{2} \\
\frac{\sinh 2r}{2} & 0 & \frac{\cosh 2r}{2} & 0 \\
0 & -\frac{\sinh 2r}{2} & 0 & \frac{\cosh 2r}{2} \\
\end{array}
\right)
\end{equation}
where $r>0$ is the squeezing parameter.
\begin{figure}
	\includegraphics[trim= 4.5cm 19cm 7cm 4cm ,clip,width=0.4\textwidth]{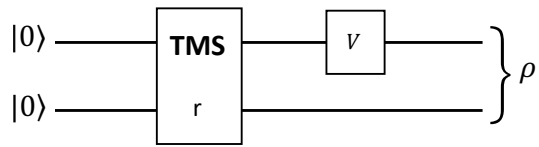} \caption{\label{exemple} 
		Example of a Gaussian state created from a two-mode vacuum state processed through a two-mode squeezer (with squeezing parameter $r>0$) and a Gaussian additive-noise channel acting on the first mode (with noise variance $V$).}
\end{figure}
By applying Eq.~(\ref{formuleRealCriterion}) where 
\begin{equation}
	\gamma_{w}=\left(
	\begin{array}{cc}
	\frac{1}{2} \left(V +e^{-2r}\right) & 0 \\
	0 & \frac{1}{2} \left(V +e^{-2r}\right) \\
	\end{array}
	\right)
\end{equation}
we find $\tr \, R(\rho)=1/(V+e^{-2r})$. According to Theorem~\ref{weak}, entanglement is thus detected if $V<1-e^{-2r}$. Clearly, for a finite squeezing parameter $r$, there exist entangled states with $1-e^{-2r}<V<1$ which are not detected. As a result, the weak realignment criterion does not always detect entanglement in this example (it becomes perfect at the limit of infinite squeezing, $r\to\infty$).

In comparison, it was shown in  \cite{Zhang} that for a matrix in the normal form Eq.~(\ref{normalform})) the trace norm of the realignment is given by Eq. (\ref{formuleZhang}). In our example, we obtain 
\begin{equation}
\parallel R(\rho)\parallel_{tr}=\frac{1}{\sqrt{(\cosh2r+2V) \cosh2r}-\sinh2r} 
\end{equation}
so that entanglement is detected if $V<~\tanh2r$.  Here again, the realignment criterion leaves some entangled states undetected (but it is more sensitive than the weak realignment criterion since $ 1-e^{-2r} <  \tanh2r$, $\forall r>0$).

Let us now symmetrize the state with the filtration procedure introduced in Sec. \ref{sectionstep1}, that is, we process the first mode (which has a larger noise variance) through a noiseless attenuation map. By inspection, we find that the optimal transmittance is $t=\tanh^2 r$ and the resulting symmetrized Gaussian state $\rho^{sym}$ admits the covariance matrix\footnote{Note that even if $V=0$ (i.e. the state already has a symmetric covariance matrix) we may still process one of its modes through the noiseless attenuation map. It simply yields another (symmetric) two-mode squeezed vacuum state with lower entanglement.}
\footnotesize
\begin{eqnarray}
\label{eq-covariance-symmetrized-state}
\gamma^{sym}&=&\frac{1}{8 (V +\cosh 2r)}\times\\
&&\begin{pmatrix}
(4 V  \cosh 2r{+}\cosh 4 r{+}3)\,\mathds{1}&(8\sinh^2r\cosh^2r)\,
\sigma_z \\ (8\sinh^2r\cosh^2r)\,\sigma_z&(4 V  \cosh 2r{+}\cosh 4 r{+}3)\,\mathds{1}
\end{pmatrix}\nonumber
\end{eqnarray}
\normalsize
where $\sigma_z=\begin{psmallmatrix}
1&0\\0&-1
\end{psmallmatrix}$. We now apply Eq.~(\ref{formuleRealCriterion}) to $\rho^{sym}$,  where 
\begin{equation}
\gamma^{sym}_{w}=\frac{1}{2}\left(
\begin{array}{cc}
\frac{V  \cosh 2r+1}{V +\cosh 2r} & 0 \\
0 & \frac{V  \cosh 2r+1}{V +\cosh 2r} \\
\end{array}
\right)
\end{equation}
which gives 
\begin{equation}
\tr \, R(\rho^{sym})=\frac{V+\cosh 2r}{1+V \cosh 2r}. 
\end{equation}
Thus, entanglement is detected if $\tr R(\rho^{sym})>1$ which is equivalent to $V<1$, for all $r$. Hence, all entangled states of the form \eqref{covEprplusbruit} are now detected. Note that 
 $\tr \, R(\rho^{sym}) = \parallel R(\rho^{sym})\parallel_{tr}$ here according to Theorem 4. Indeed, the covariance matrix \eqref{eq-covariance-symmetrized-state} is in the form \eqref{normalform+a=b} with $c>0$ and $d<0$, so we have reached a Schmidt-symmetric state.
 
As a consequence, we have confirmed that the entanglement detection for Gaussian states is improved if one symmetrizes the state before applying the weak realignment criterion. In particular, in this specific example, the weak realignment criterion is as strong as the original realignment criterion with symmetrization since $\tr \, R(\rho^{sym})=\parallel R(\rho^{sym})\parallel_{tr}$ and even stronger than the realignment criterion without symmetrization based on $\parallel R(\rho)\parallel_{tr}$. Moreover, applying the symmetrization procedure and computing the trace of the realigned state (via the determinant of the restricted covariance matrix) are much easier than computing the trace norm of the realigned state (as developed in \cite{Zhang}). In Fig. \ref{EPRplusBruit} (upper panel), we illustrate the fact that the trace and trace norm of the realigned state can be increased by the filtration procedure (the value without filtration is found when $t=1$). We notice that, although the optimal value of the transmittance $t=\tanh^2(r)$ allows for the detection of entanglement, there are actually many other values of $t$ that allow for such a detection too. Moreover, it seems that the symmetrized state $t=\tanh^2(r)$ is not necessarily the best way of filtering the state of this example as it does not give the highest possible value of the trace of $R(\rho)$.
 

\begin{figure}
	  \subfloat[][Noise added on the first mode -- attenuation map]{%
	\includegraphics[width=0.45\textwidth]{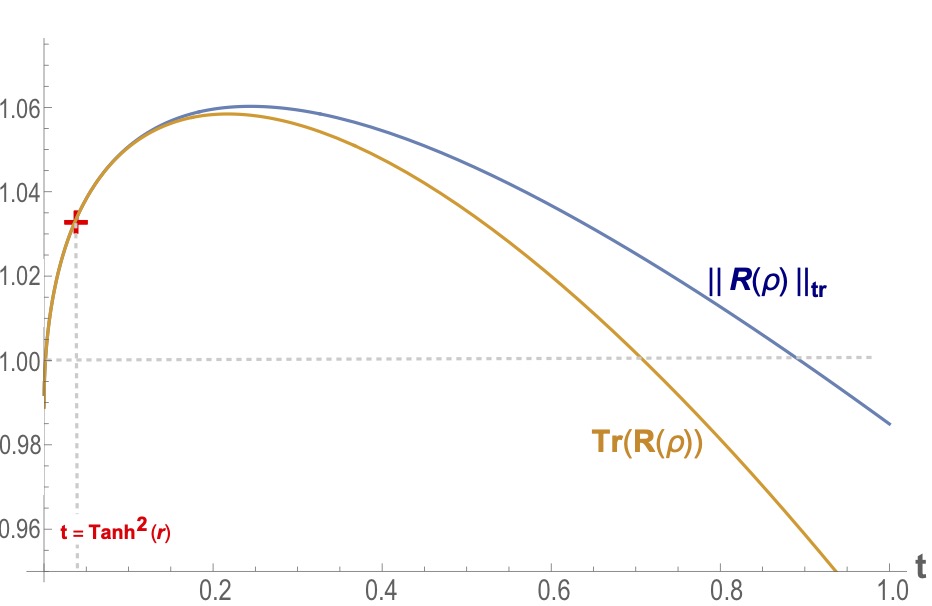} 
	}
	\hfill
	\subfloat[][Noise added on the second mode -- amplification map]{%
	\includegraphics[width=0.45\textwidth]{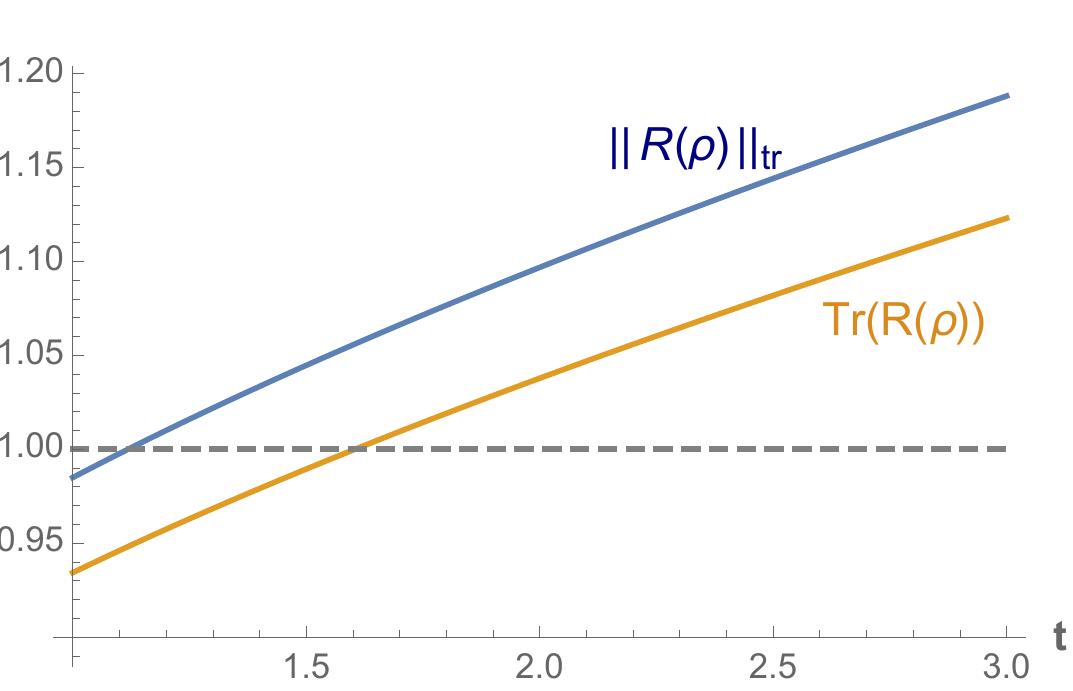} 
	}
	\caption{\label{EPRplusBruit} Comparison of the trace $\tr \, R(\rho)$ and trace norm $\parallel R(\rho)\parallel_{tr}$ of the realigned state for a two-mode squeezed vacuum state processed through a Gaussian additive-noise channel (we choose $r=0.2$ and $V=0.4$). \\
		(a) The noise is added on the first mode and filtering consists in processing this mode via a noiseless attenuation map. The red cross shows the point at $t=\tanh^2(r)$ where the covariance matrix has been symmetrized and $\tr \, R(\rho)=\parallel R(\rho)\parallel_{tr}$. \\
		(b) The noise is added on the second mode and filtering works by processing the first mode via a noiseless amplification map. The values of the trace and the trace norm coincide when $t=\frac{1}{\tanh^2r}\approx25.7$.}
\end{figure}

As a second example, let us start with the same two-mode squeezed vacuum state but  add noise on the second mode instead of the first. The covariance matrix reads
 \begin{equation}
 \gamma=\left(
 \begin{array}{cccc}
 \frac{\cosh 2r}{2} & 0 & \frac{\sinh 2r}{2} & 0 \\
 0 & \frac{\cosh 2r}{2} & 0 & -\frac{\sinh 2r}{2} \\
 \frac{\sinh 2r}{2} & 0 & V +\frac{\cosh 2r}{2} & 0 \\
 0 & -\frac{\sinh 2r}{2} & 0 &V + \frac{\cosh 2r}{2} \\
 \end{array}
 \right).
 \end{equation}
As before $\tr \, R(\rho)=1/(V+e^{-2r})$ so the weak realignment criterion alone does not detect all entangled states. To improve on this, we could of course apply the noiseless attenuation map on the second mode, which would give the exact same results. Alternatively, we may explore another filtration procedure, which consists in applying the noiseless amplification map on the first mode. As can be seen in Fig. \ref{EPRplusBruit} (lower panel), the values of the trace and trace norm increase with $t$, and if $t$ is chosen big enough, we detect entanglement. In order to symmetrize the covariance matrix, we would need a noiseless amplifier of transmittance $t=\frac{1}{\tanh^2r}$, which is a limiting case that would yield a two-mode squeezed vacuum state with infinite squeezing. For this optimal value of $t$, we have $\tr \, R(\rho_{sym})=\frac{1}{V}$ and thus entanglement is always detected when $V<1$. Hence, here again, the weak realignment criterion assisted with filtration allows us to detect all entangled states.

 \subsection{Random two-mode Gaussian states}
 Let us consider two other examples of random two-mode Gaussian states with their covariance matrices written in the normal form:
\scriptsize
\begin{eqnarray}
\gamma_1=\left(
\begin{array}{cccc}
1.46 & 0 & 0.83 & 0 \\
0 & 1.46 & 0 & -0.23 \\
0.83 & 0 & 0.80 & 0 \\
0 & -0.23 & 0 & 0.80 \\
\end{array}
\right),\,   \nonumber \\
\gamma_2=\left(
\begin{array}{cccc}
1.29 & 0 & -0.76 & 0 \\
0 & 1.29 & 0 & 0.44 \\
-0.76 & 0 & 0.83 & 0 \\
0 & 0.44& 0 & 0.83 \\
\end{array}
\right).
\end{eqnarray}
\normalsize
These states are not PPT so they  are entangled. These examples are interesting because in both cases $\parallel R(\rho)\parallel_{tr}>1$ but $\tr \, R(\rho)<1$, so entanglement is detected by the realignment criterion but not by its weak formulation. We thus need to apply the filtration procedure in order to enhance the detection with the weak realignment criterion. In Fig. \ref{Examples}, we show the evolution of $\tr \, R(\rho)$ as a function of the transmittance $t$ of the noiseless attenuation map applied on the first mode ($t=1$ corresponds to the initial value when no filtration is applied). The red cross indicates the exact point when the covariance matrix has been symmetrized. In the first example (see Fig. \ref{Examples}a), the filtration procedure works well and many values of $t$ allow us to detect entanglement. In particular, the entanglement is detected at the optimal value of $t$ (note that the trace and trace norm do not exactly coincide there, which witnesses the fact that the symmetrized state is not exactly a Schmidt-symmetric state). In the second example (see Fig. \ref{Examples}b), however, filtration alone is not sufficient and entanglement is never detected by the weak realignment criterion. Even if filtration is performed by applying a noiseless amplifier map on the second mode, we observe the same results. Nevertheless, entanglement can still be detected if we apply a local rotation (a $\pi$ phase shift on one of the two modes which has the effect to flip the sign of the $c$ and $d$ elements in the normal form of the covariance matrix)  prior to the filtration procedure, which makes the covariance matrix look similar to the first example. This is shown by the dashed green curve on Fig. \ref{Examples}b. 
Furthermore, by applying an appropriate local squeezing on the second mode of the state after the noiseless attenuator on the first mode, we may always reach a Schmidt-symmetric state (provided $c$ and $d$ have opposite signs in the covariance matrix \eqref{normalform} of the initial state, otherwise the state is anyway separable). This indicates that applying a suitable local phase shift followed by a suitable noiseless attenuator (or amplifier) and finally a suitable local squeezer yields a filtration procedure that always allows the detection of entanglement for a two-mode Gaussian state.

%

%

Note that  we cannot plot the evolution of $\parallel R(\rho)\parallel_{tr}$ as a function of $t$ in Fig. \ref{Examples} (in contrast with Fig. \ref{EPRplusBruit}) since the covariance matrix after filtration is not anymore in the form (\ref{normalform}). The blue dashed line represents its initial value before the filtration is applied.

\begin{figure}
		  \subfloat[][Covariance matrix $\gamma_1$]{%
	\includegraphics[width=0.45\textwidth]{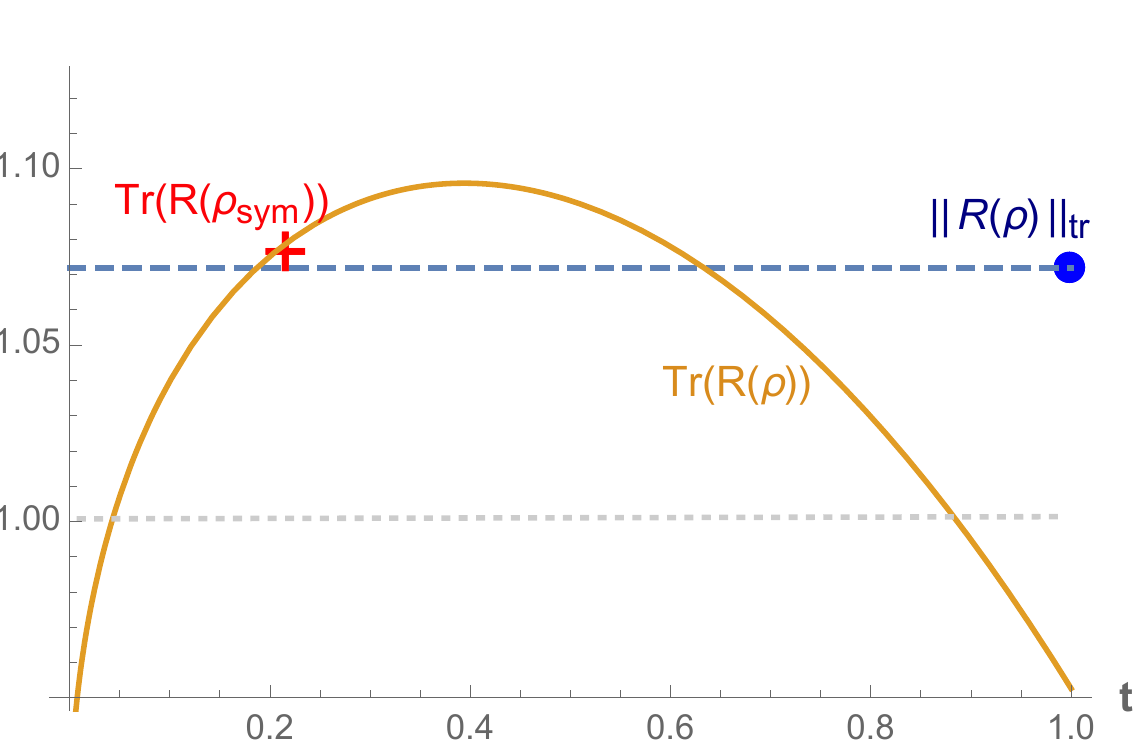}
	}
	\hfill
	\subfloat[][Covariance matrix $\gamma_2$]{%
	\includegraphics[width=0.45\textwidth]{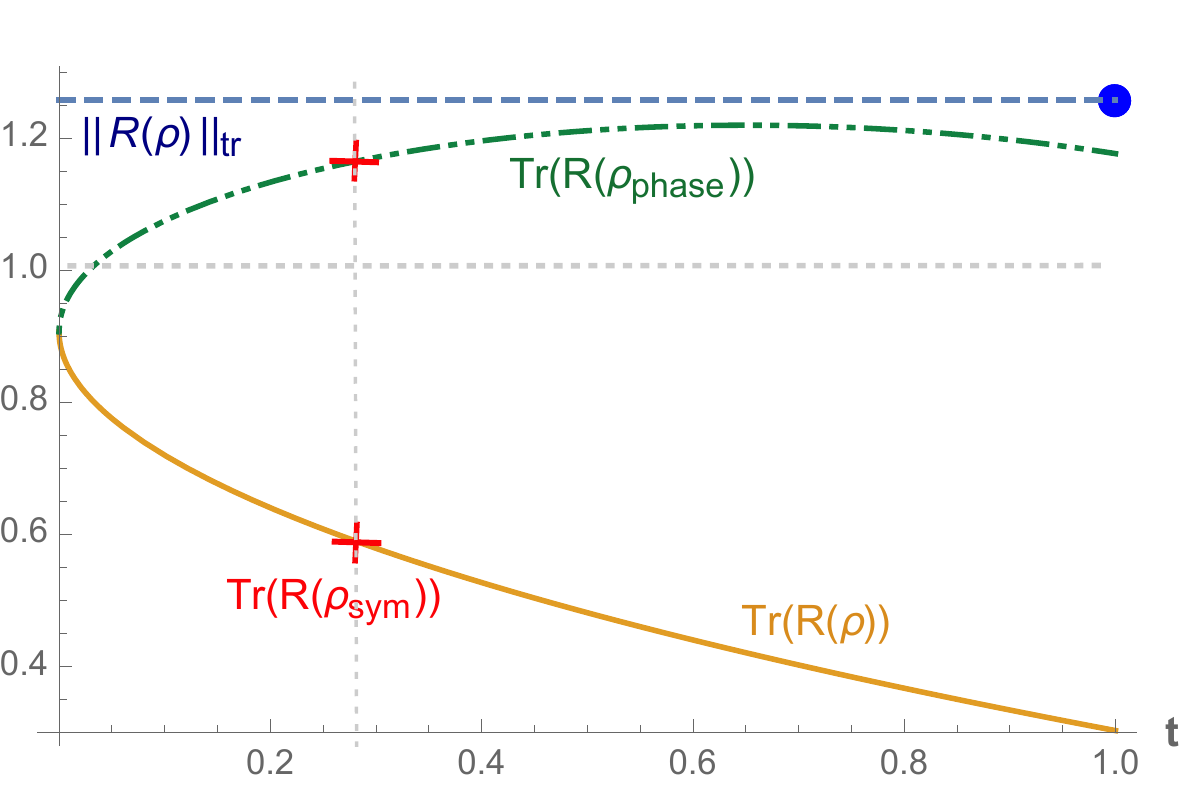} 
	}
	\caption{\label{Examples} Evolution of the trace $\tr \, R(\rho)$ of the realigned state as a function of $t$ for two-mode Gaussian states with covariance matrices (a) $\gamma_1$ and (b) $\gamma_2$ after filtering (noiseless attenuation on the first mode). Both examples are entangled states. The blue dashed line represents $\parallel R(\rho)\parallel_{tr}$ before the filtration and the red cross shows the value of $\tr \, R(\rho_{sym})$ when the covariance matrix has been symmetrized. The green dashed curve shows the evolution of the trace $\tr \, R(\rho_{phase})$, where $\rho_{phase}$ is the state obtained by applying a $\pi$ phase shift prior to the filtration.}
\end{figure}

 \subsection{Examples of $2\times 2$ NPT Gaussian states}

Let us now move on to examples of $2\times 2$ Gaussian states (in which case the PPT criterion is not any more necessary and sufficient). We extend the example of Sec. \ref{ExampleEPR} by considering that Alice and Bob share two instances of a  two-mode squeezed vacuum state with added noise. The covariance matrix is thus given by
\footnotesize
\begin{equation}
\gamma_{EPR}{=}\left(
\begin{array}{cccc}
\left(V {+}\frac{\cosh 2r}{2} \right)\, \mathds{1}& 0 & \frac{\sinh 2r}{2} \,\sigma_z& 0 \\
0 &\left(V {+}\frac{\cosh 2r}{2} \right)\, \mathds{1}& 0 & \frac{\sinh 2r}{2} \,\sigma_z\\
\frac{\sinh 2r}{2} \,\sigma_z& 0 &\frac{\cosh 2r}{2} \, \mathds{1}& 0 \\
0 & \frac{\sinh 2r}{2}\,\sigma_z & 0 &\frac{\cosh 2r}{2} \, \mathds{1} \\
\end{array}
\right).
\end{equation}
\normalsize
This state is always detected by the PPT separability criterion. We can also add some rotations on Bob's modes in order to get another state whose covariance matrix is given by $\gamma'_{EPR}=R(\theta,\tau)\, \gamma_{EPR} \, R^T(\theta,\tau)$ with
\begin{eqnarray}
 R(\theta,\tau)&=&\begin{pmatrix}
\mathds{1}_{4\times 4}&0&0&0\\0&\cos \theta&\sin\theta&0\\0&-\sin\theta&\cos\theta&0\\0&0&0&\mathds{1}_{2\times2}
\end{pmatrix}\\
&\times&\begin{pmatrix}
\mathds{1}_{4\times 4}&0&0&0&0\\
0&\sqrt{\tau}&0&-\sqrt{1-\tau}&0\\
0&0&\sqrt{\tau}&0&-\sqrt{1-\tau}\\
0&\sqrt{1-\tau}&0&\sqrt{\tau}&0\\
0&0&\sqrt{1-\tau}&0&\sqrt{\tau}
\end{pmatrix}  .  \nonumber
\end{eqnarray}
This rotated state is always entangled and detected by the PPT criterion. The entanglement detection effected by the weak realignment criterion is, however, depending on the values of $\theta$ and $\tau$ as follows.
  \begin{itemize}
\item If $\theta =0$ and $\tau=1$ we have  $\gamma'_{EPR}=\gamma_{EPR}$ and the calculations are exactly the same as in Sec. \ref{ExampleEPR} (but everything is squared because we now have two states). It means in particular that $\tr \, R(\rho_{EPR})=\frac{1}{\left(e^{-2 r}+V\right)^2}$ is not always greater than 1, but if we applied a suitable filtration with $t=\tanh^2(r)$, entanglement becomes always detected.

\item If $\theta=\pi$ and regardless of  the value of $\tau$, we have $\tr \, R(\rho'_{EPR})=\frac{1}{1+V^2+2V\cosh 2r}$ which is always smaller than 1. Entanglement is thus never detected. Note that in this particular case, the filtration does not improve the value of $\tr \, R(\rho')$ even if we try to add a rotation before the filtration. The key point is that this state does not have EPR-like correlations 

\item If $\theta=0$ and regardless of  the value of $\tau$, we have $\tr \, R(\rho'_{EPR})=\frac{1}{(\cosh 2r -\sqrt{\tau}\sinh 2r +V)^2}$. In some cases, entanglement is detected without any filtration. In some other cases, entanglement is not straightforwardly detected, but the filtration helps in the detection. For example, if $r=1$, $V=0.8$, and $\tau=0.9$,  then $\tr \, R(\rho'_{EPR})\approx0.8<1$ and entanglement is not detected as such. However, if we apply the filtration procedure, we see in Fig. \ref{Example3}  (upper panel) that there are many values of $t$ that enable entanglement detection. 

\begin{figure}
	\includegraphics[width=0.35\textwidth]{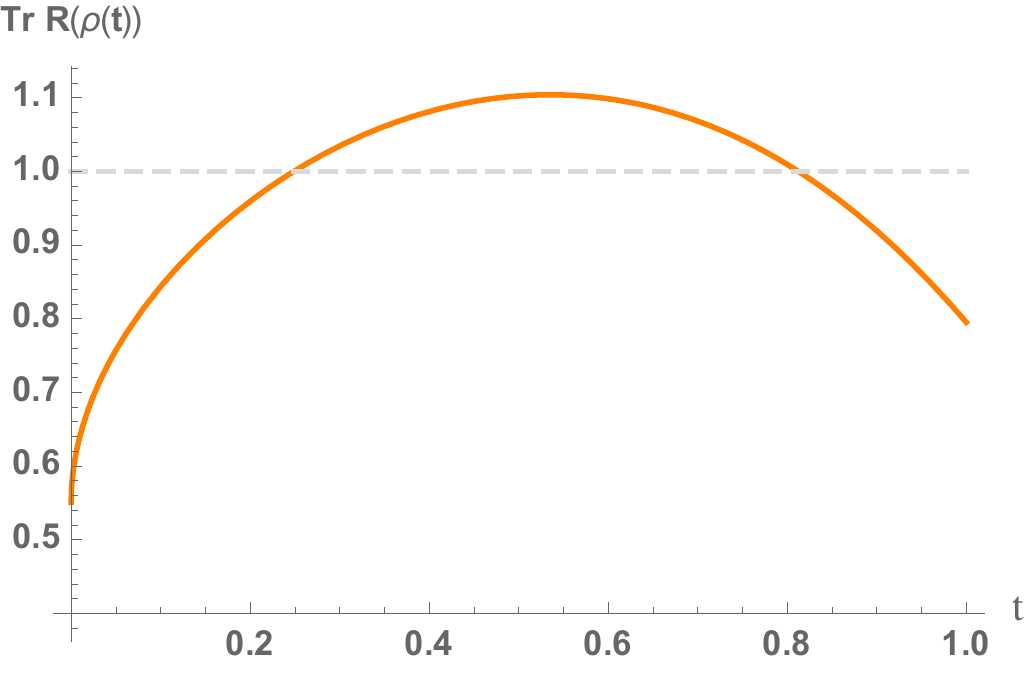} 
	\caption{\label{Example3} Evolution of the trace $\tr \, R(\rho'_{EPR})$ of the realigned state after filtration as a function of $t$ for $r=1,V=0.8,\tau=0.9$ and $\theta=0$.The entanglement is detected in the interval of $t$ values such that the trace exceeds 1.
	}
\end{figure} 
 \end{itemize}

\section{Conclusions}


We have introduced a weak formulation of the realignment criterion based on the trace of the realigned state $R(\rho)$, which has the advantage of being much easier to compute than the original formulation of the realignment criterion, especially in higher dimensions. It has a simple physical implementation as computing $\tr \, R(\rho)$ is equivalent to measuring the $\ket{\Omega}$ component of state $\rho$ via linear optics and homodyne measurements. Moreover, for states in the Schmidt-symmetric form, both realignment criteria --- the weak and the original formulations --- are equivalent. We focused especially on Gaussian states and showed that applying a suitable filtration procedure prior to applying the weak realignment criterion often allows for a better entanglement detection. In particular, we have explored a filtration based on noiseless amplification or attenuation, which is an invertible operation that transforms the state into a symmetrized form such that the entanglement detection is enhanced (this procedure may even surpass the original realignment criterion while it is simpler). We have provided examples of the application of this procedure for various  $1\times 1$ and $2 \times 2$ Gaussian states. These examples illustrate the power of the method (it can be made to detect all  entangled $1\times 1$ Gaussian states), even though we have found cases where it leaves the entanglement of $2 \times 2$ states undetected.

A question that we leave open in this work is whether the weak realignment criterion assisted with suitable prior filtration can be stronger than the PPT criterion to the degree that it can detect bound entangled states. The weak realignment criterion is weaker than both the original realignment and PPT criteria, which are two incomparable criteria (except for states in the symmetric subspace, where they coincide). Hence, as such, it cannot detect bound entanglement.  We have not been able to find instances where adding filtration allowed us to detect bound entangled states, although it should in principle be possible to bring the state close enough to a Schmidt symmetric state so that bound entanglement is detected. Note, however, that this can only be the case if the original realignment criterion detects the entanglement of the state, that is, if the Schmidt-symmetric state that is approached via filtration is not within the symmetric subspace (otherwise, the weak realignment criterion tends to the realignment criterion which itself coincides with the PPT criterion, so no bound entanglement can be detected). 
This puts severe constraints on where to seek  the detection of bound entanglement using the weak realignment criterion assisted with filtration.

As a future work, it would also be interesting to explore other possible filtration procedures in order to further improve the entanglement detection in higher dimensions. Alternatively, another interesting goal would be to find a physically implementable protocol for the original realignment criterion and not its weak form (that is, for evaluating the trace norm of $R(\rho)$ instead of its trace with optical components).

\medskip
\noindent {\it Acknowlegments}: 
AA. thanks  Cheng-Jie Zhang for helpful discussions.
This work was supported by the F.R.S.- FNRS Foundation under Project No. T.0224.18, by the FWF under Project No. J3653-N27, and by the European Commission under project ShoQC within the ERA-NET Cofund Programme in Quantum Technologies  (QuantERA). AH also acknowledges financial support from the F.R.S.-FNRS Foundation and from the Natural Sciences and Engineering Research Council of Canada (NSERC).
\bigskip


\appendix

\section{Proof of Theorem \ref{theoreal1}}
\label{Appendix0}
We give here the proof of Theorem \ref{theoreal1}.
	Let us construct an entanglement witness $\mathcal{W}$, that is, an observable with a positive expectation value on all separable states. Let us define $\mathcal{W}=\mathds{1}-\sum_i^{r}A_i\otimes B_i$ and let us check that $\tr (\rho_{sep}\mathcal{W})\geq0$ where $\rho_{sep}=(\ket{a}\otimes\ket{b})(\bra{a}\otimes\bra{b})$ is a separable (product) state. First we remark that
	\begin{eqnarray}
	\tr (\rho_{sep}\mathcal{W})&=&(\bra{a}\otimes\bra{b})\mathcal{W}(\ket{a}\otimes\ket{b})\\
	&=&1-\sum_{i}^{r}\mean{a|A_i|a}\mean{b|B_i|b}\nonumber\\
	&\geq&1-\sqrt{\sum_{i}^{r}|\mean{a|A_i|a}|^2}\sqrt{\sum_{i}^{r}|\mean{b|B_i|b}|^2}\nonumber
	\end{eqnarray}
	where we used the Cauchy-Schwarz inequality in the last step. Now, since the $\{A_i\}$  form a basis, we can write $\ketbra{a}{a}=\sum_j\alpha_j A_j$ where $\alpha_j=\mean{a|A_i|a}$, and similarly for $\ketbra{b}{b}$. This allows us to write
	\begin{eqnarray}
	1&=&\parallel \ketbra{a}{a}\parallel^2=\tr(\ketbra{a}{a}(\ketbra{a}{a})^\dag)=\tr(\sum_{ij}\alpha_iA_i\alpha_j^*A_j^\dag)\nonumber\\
	&=&\sum_{ij}\alpha_i\alpha_j^*\tr(A_iA_j^\dag)=\sum_i|\alpha_i|^2=\sum_i|\mean{a|A_i|a}|^2
	\end{eqnarray}
	and similarly $\sum_i|\mean{b|B_i|b}|^2=1$ so that $\tr(\rho_{sep}\mathcal{W})\geq 0$. Thus, $\mathcal{W}$ is indeed an entanglement witness as any separable state is expressed as a convex mixture of states of the form $\rho_{sep}$. Let us now check under which condition it allows for entanglement detection. In other words, what is the condition to have $\tr(\rho\mathcal{W})<0$? Consider a state $\rho$ written in its operator Schmidt decomposition. Then,
	\begin{eqnarray}
	\tr(\rho\mathcal{W})&=&\mathds{1}-\tr\left(\sum_{ij}^{r}\lambda_i A_i A_j\otimes B_i B_j\right)\nonumber\\
	&=&\mathds{1}-\sum_{ij}^{r}\lambda_i \tr(A_iA_j)\tr(B_iB_j)\nonumber\\
	&=&1-\sum_{i}^{r}\lambda_i .
	\end{eqnarray}
	Entanglement is thus detected when $\sum_{i}^{r}\lambda_i >1$ which completes the proof.

%
%

\section{Formulation of the $\ket{\Omega}$ state }
\label{AppendixA}


We prove here that the state $\ket{\Omega}$ can be reexpressed as $\sqrt{\pi} \, U_{BS}^\dagger \ket{0}_{x_1}\ket{0}_{p_2}$ where $U_{BS}$ is the unitary of a 50:50 beam splitter. By definition, it is expressed in the Fock basis as $\ket{\Omega}=\sum_n\ket{n}\ket{n}$. Thus, if $\ket{x}$ and $\ket{y}$ are position states, we have
\begin{eqnarray}
\bra{x}\bra{y}\Omega\rangle&=& \sum_n \langle x | n\rangle \langle y | n\rangle \nonumber \\
&=&  \sum_n \langle x | n\rangle \langle n | y\rangle  \nonumber \\
&=&  \langle x | y\rangle = \delta(x-y) ,  \label{eq-omega-position}
\end{eqnarray}
so that $\ket{\Omega}$ can be written in the position basis as
\begin{equation}
\ket{\Omega}=\int dx \, dy \, \delta(x-y) \ket{x}\ket{y}=\int dx \, \ket{x}\ket{x}.
\end{equation}
Since the action of the 50:50 beam splitter unitary $U_{BS}$ on the position eigenstates is defined as
\begin{equation}
U_{BS}\ket{x}\ket{y}=\left|\frac{x-y}{\sqrt{2}}\right\rangle  \left|\frac{x+y}{\sqrt{2}}\right\rangle 
\end{equation} 
we have,
\begin{eqnarray}
\bra{x}\bra{y}U_{BS}^\dagger \ket{0}_{x_1}\ket{0}_{p_2}    
&=&\left\langle\frac{x-y}{\sqrt{2}}\right|\left\langle\frac{x+y}{\sqrt{2}}\right|
 \ket{0}_{x_1}\ket{0}_{p_2}  \nonumber\\
 &=&\left\langle\frac{x-y}{\sqrt{2}}\Big|x=0\right\rangle\left\langle\frac{x+y}{\sqrt{2}}\Big|p=0\right\rangle\nonumber\\
 &=&\delta\left(\frac{x-y}{\sqrt{2}}\right)\frac{1}{\sqrt{2\pi}}\nonumber\\
&=&\delta(x-y)  /  \sqrt{\pi}
\end{eqnarray}
where we have used the fact that $\mean{x|y}=\delta(x-y)$ and $\mean{x|p}=\frac{1}{\sqrt{2\pi}}e^{ipx}$. Comparing with Eq. \eqref{eq-omega-position}, this completes the proof that $\ket{\Omega}=\sqrt{\pi} \, U_{BS}^\dagger \ket{0}_{x_1}\ket{0}_{p_2}$.

\section{Computation of Eq.~(\ref{probGauss})}
\label{AppendixProb}
	We show how to directly compute Eq.~(\ref{probGauss}) for a two-mode state.  We have to compute the following integral where we set $x_1=p_2=0$:
	\begin{equation}
	\begin{aligned}
	\bra{x=0}&\bra{p=0}\rho'\ket{x=0}\ket{p=0}\\
	&=\int\,dx_2dp_1 W_{\rho'}(0,p_1,x_2,0)\\
	&=\frac{1}{(2\pi)^2\sqrt{\det\gamma'}}\int\,dx_2dp_1 \,e^{-\frac{1}{2}\left(\begin{smallmatrix} 0&p_1&x_2&0 	\end{smallmatrix}\right)(\gamma')^{-1}\left(\begin{smallmatrix}	0\\p_1\\x_2\\0	\end{smallmatrix}\right)}\\
	&=\frac{1}{(2\pi)^2\sqrt{\det\gamma'}}\int\,dx_2dp_1 \,e^{-\frac{1}{2}\left(\begin{smallmatrix} x_2&p_1 	\end{smallmatrix}\right)\Gamma\left(\begin{smallmatrix}	x_2\\p_1	\end{smallmatrix}\right)}\\
	&=\frac{1}{(2\pi)^2\sqrt{\det\gamma'}}\frac{2\pi}{\sqrt{\det\Gamma}}\\
	&=\frac{1}{2\pi\sqrt{\det\gamma'}}\sqrt{\frac{\det\gamma'}{\det\gamma_{w}}}=\frac{1}{2\pi\sqrt{\det\gamma_{w}}}
	\end{aligned}
	\end{equation}
	where $\Gamma$ is a $2\times 2$ matrix with elements given by $\Gamma_{1,1}=(\gamma')^{-1}_{3,3}$, $\Gamma_{2,2}=(\gamma')^{-1}_{2,2}$ and $\Gamma_{1,2}=\Gamma_{2,1}=(\gamma')^{-1}_{2,3}$ 

\section{Physical interpretation of the symmetrization procedure for a Gaussian state}
\label{AppendixB}

We present here an alternative way of computing the covariance matrix of the output state of a noiseless attenuation channel. To do so, we use the fact that the attenuation channel can be represented by a beam splitter followed by a postselection on the vacuum (see Fig.~\ref{symm}).

To filter the state, we process the modes of the subsystem with the higher variance (that is the higher value of the determinant of the reduced covariance matrix $A$ or $B$) through a noiseless attenuation channel and  then postselect the output conditionally to measuring the vacuum on the ancillary modes. Since it is a Gaussian channel, the output remains Gaussian.  Processing the state through this channel will have for effect to lower the variance of the mode that traveled through the channel. The output state is the symmetrized Gaussian state $\rho^{sym}$. We chose the attenuator factor (that is the transmittance $t$ of the beam splitter) so that the variance of both modes of $\rho^{sym}$ are equal that is $\det A= \det B$. In terms of covariance matrix, the procedure is as follows.

Let us have an $n \times n$ Gaussian state $\rho$ with covariance matrix (\ref{covmatrix}) and let us assume $ \det A\geq \det B$ with no loss of generality.
We then add $n$ vacuum state to the system. The new covariance matrix thus reads
\begin{equation}
\gamma_{\ket{0}^{\oplus n}+\rho}=\begin{pmatrix}
\frac{1}{2} \mathds{1}_{2n}&0\\0&\gamma
\end{pmatrix}.
\end{equation}
We now apply the transformation $\mathcal{S}\oplus\mathds{1}$ where $\mathcal{S}$ is the beam splitter transformation 
\begin{equation}
\mathcal{S}\oplus\mathds{1}=
\begin{pmatrix}
\sqrt{t}\mathds{1}_{2n}&-\sqrt{1-t}\mathds{1}_{2n}&0\\
\sqrt{1-t}\mathds{1}_{2n}&\sqrt{t}\mathds{1}_{2n}&0\\
0&0&\mathds{1}_{2n}
\end{pmatrix}.
\end{equation}
to the covariance matrix $\gamma_{\ket{0}^{\oplus n}+\rho}$:
\begin{equation}
\begin{aligned}
\gamma_{\mathcal{S}\oplus\mathds{1}}=\mathcal{S}\oplus\mathds{1}\,\gamma_{\ket{0}^{\oplus n}+\rho}\,\mathcal{S}^\dag\oplus\mathds{1}=\begin{pmatrix}
\mathcal{A}&\mathcal{C}^T\\\mathcal{C}&\mathcal{B}
\end{pmatrix}
\end{aligned}.
\end{equation}
Finally, we reduce the covariance matrix conditionally to measuring the vacuum on the first mode, that is \cite{fiurasek,Giedke}
\begin{equation}
\gamma^{sym}=\mathcal{B}-\mathcal{C}\left(\mathcal{A}+\frac{1}{2}\mathds{1}\right)^{-1}\mathcal{C}^T = \begin{pmatrix}
\mathcal{A'}&\mathcal{C'}^T\\\mathcal{C'}&\mathcal{B'}
\end{pmatrix}.
\label{rhosym}
\end{equation}
At this stage we obtained a new covariance matrix which depends on $t$. If the filtration procedure is such that we want to symmetrize the covariance matrix, we need to make sure that the determinant of the covariance matrices of both subsystems are equal ($\det \mathcal{A'}=\det \mathcal{B'})$. 

\subsection*{Explicit calculation for the two-mode case}

Let us do the explicit calculations to obtain the symmetrized covariance matrix of a two-mode Gaussian state initially expressed in its normal form \cite{duan},
\begin{equation}
\gamma_\rho=\begin{pmatrix}
a&0&c&0\\0&a&0&d\\c&0&b&0\\0&d&0&b
\end{pmatrix}.
\end{equation}
Any covariance matrix of a two-mode state can be transformed into this form
by applying local linear unitary operations which are combinations of squeezing
transformations and rotations. These operations do not influence the separability
of the state, and are thus always allowed when studying entanglement. Note that we assume $a\geq b$ with no loss of generality.
We first add the vacuum state to the system. The new covariance matrix  reads
\begin{equation}
\gamma_{\ket{0}+\rho}=\begin{pmatrix}
1/2&0&0&0&0&0\\0&1/2&0&0&0&0\\
0&0&a&0&c&0\\0&0&0&a&0&d\\0&0&c&0&b&0\\0&0&0&d&0&b
\end{pmatrix}.
\end{equation}
We then apply the transformation $\mathcal{S}\oplus\mathds{1}$
to the covariance matrix $\gamma_{\ket{0}+\rho}$ to obtain
\begin{equation}
\begin{aligned}
\gamma_{\mathcal{S}\oplus\mathds{1}}=\mathcal{S}\oplus\mathds{1}\,\gamma_{\ket{0}+\rho}\,\mathcal{S}^\dag\oplus\mathds{1}=\begin{pmatrix}
\mathcal{A}&\mathcal{C}^T\\\mathcal{C}&\mathcal{B}
\end{pmatrix}
\end{aligned}
\end{equation}
with
\begin{eqnarray}
\mathcal{A}&=&\left(
\begin{array}{cc}
-t a+a+\frac{t}{2} & 0 \\
0 & -t a+a+\frac{t}{2} \\
\end{array}
\right)\\
\mathcal{B}&=&\left(
\begin{array}{cccc}
\left(a-\frac{1}{2}\right) t+\frac{1}{2} & 0 & c \sqrt{t} & 0 \\
0 & \left(a-\frac{1}{2}\right) t+\frac{1}{2} & 0 & d \sqrt{t} \\
c \sqrt{t} & 0 & b & 0 \\
0 & d \sqrt{t} & 0 & b \\
\end{array}
\right)\nonumber\\
\mathcal{C}&=&\left(
\begin{array}{cc}
\frac{1}{2} (1-2 a) \sqrt{-(t-1) t} & 0 \\
0 & \frac{1}{2} (1-2 a) \sqrt{-(t-1) t} \\
-c \sqrt{1-t} & 0 \\
0 & -d \sqrt{1-t} \\
\end{array}
\right).\nonumber
\end{eqnarray}
Finally, we reduce the covariance matrix conditionally to measuring the vacuum on the first mode, that is
\begin{widetext}
\begin{equation}
\begin{aligned}
\gamma^{sym}&=\mathcal{B}-\mathcal{C}\left(\mathcal{A}+\frac{1}{2}\mathds{1}\right)^{-1}\mathcal{C}^T   =\left(
\begin{array}{cccc}
\frac{t-2 a (t+1)-1}{4 a (t-1)-2 (t+1)} & 0 & \frac{2 c \sqrt{t}}{-2 a (t-1)+t+1} & 0 \\
0 & \frac{t-2 a (t+1)-1}{4 a (t-1)-2 (t+1)} & 0 & \frac{2 d \sqrt{t}}{-2 a (t-1)+t+1} \\
\frac{2 c \sqrt{t}}{-2 a (t-1)+t+1} & 0 & \frac{2 (t-1) c^2}{-2 a (t-1)+t+1}+b & 0 \\
0 & \frac{2 d \sqrt{t}}{-2 a (t-1)+t+1} & 0 & \frac{2 (t-1) d^2}{-2 a (t-1)+t+1}+b \\
\end{array}
\right).
\end{aligned}
\label{rhosymSpecific}
\end{equation}
\end{widetext}
The covariance matrix will be symmetrized providing that the determinant of the covariance matrices of both subsystems are equal ($\det A=\det B)$, meaning 
\begin{eqnarray}
\left(\frac{t-2 a (t+1)-1}{4 a (t-1)-2 (t+1)}\right)^2 &=&\left(\frac{2 (t-1) c^2}{-2 a (t-1)+t+1}+b\right)\nonumber\\
&\times&\left(\frac{2 (t-1) d^2}{-2 a (t-1)+t+1}+b\right).\nonumber\\
\end{eqnarray}
Solving this equation for $t$ gives the transmissivity of the beam splitter necessary to obtain a  Gaussian state with a symmetric covariance matrix.


\begin{thebibliography}{99}
	
	\bibitem{horodecki} R. Horodecki, P. Horodecki, M. Horodecki ans K. Horodecki. Rev. Mod. Phys. \textbf{81}:865-942 (2009).
	
	\bibitem{Guhne2009} O. Gühne and G. Tóth. Phys. Rep. \textbf{474}:1-75 (2009).
	
	
	\bibitem{peres} A. Peres. Phys. Rev. Lett., \textbf{77} 1413  (1996).
\bibitem{horodecki1996}	R. Horodecki, P. Horodecki, and M. Horodecki. Phys. Lett. A, \textbf{210} 377 (1996).
	
			\bibitem{duan}L. M. Duan, G. Giedke, J. I. Cirac and  P. Zoller.	Phys. Rev. Lett. {\bf 84} 2722 (2000).
		\bibitem{00Simon} R. Simon. Phys. Rev. Lett. \textbf{84}, 12, 2726 (2000).
	\bibitem{WernerWolf} R. F. Werner and M. M. Wolf, Phys. Rev. Lett. \textbf{86}, 3658 (2001).
\bibitem{Serafini} A. Serafini, G. Adesso, and F. Illuminati, Phys. Rev. A \textbf{71},
032349 (2005).

\bibitem{horodecki98} M. Horodecki, P. Horodecki, and R. Horodecki,  Phys. Rev.
Lett. \textbf{80}, 5239 (1998).

\bibitem{shchukin}
E. Shchukin and W. Vogel, Phys. Rev. Lett. \textbf{95}, 230502(2005).
\bibitem{walborn} S. P. Walborn, B. G. Taketani, A. Salles, F. Toscano, and R. L. de Matos Filho, Phys. Rev. Lett. \textbf{103}, 160505(2009).

\bibitem{Lami} L. Lami, A Serafini and G. Adesso. New J. Phys. \textbf{20} 023030 (2018).

\bibitem{Mardani} Y. Mardani, A. Shafiei, M. Ghadimi, and M. Abdi
Phys. Rev. A \textbf{102}, 012407(2020).

\bibitem{Mihaescu}T. Mihaescu, H. Kampermann, G. Gianfelici, A. Isar, D. Bruss
New J. Phys.  \textbf{22}, 123041 (2020).


 \bibitem{Rudolph} O. Rudolph. Further results on the cross norm criterion for separability
Quant. Inf. Proc. \textbf{4} 219 (2005).
 \bibitem{Chen} K. Chen and L. A. Wu.  Quant. Inf. Comp. \textbf{3}(3) 193 (2003).

		\bibitem{Zhang} C. Zhang, S. Yu, Q. Chen and C.H. Oh. Phys. Rev. Lett. {\bf 111}, 190501 (2013).
				\bibitem{NielsenChuang} M. A. Nielsen and I. L. Chuang, Cambridge University Press, 10th anniversary edition (2010), ISBN \textbf{978-1-107-00217-3}.

	\bibitem{Nielsen}
	M. A. Nielsen, Phys. Rev. Lett. \textbf{83}, 436 (1999)

	\bibitem{Peres93} A. Peres. Quantum Theory: Concepts and Methods (Dordrecht: Kluwer Academic Publishers) (1993).
				
	\bibitem{Choi}		M.D. Choi. Linear algebra and its applications, \textbf{10}(3), 285-290 (1975).
				
				\bibitem{Jamio}		A. Jamio\l kowski. Reports on Mathematical
				Physics, \textbf{3}(4), 275-278 (1972).

\bibitem{zhangzhang} C.-J. Zhang, Y.-S. Zhang, S. Zhang, G.-C. Guo. Phys. Rev. A \textbf{77}, 060301(R) (2008).
\bibitem{johnston}N. Johnston. Lecture notes. Entanglement
Detection, \url{http://www.njohnston.ca/ed.pdf} (2014).


	\bibitem{WolfThesis} M.M. Wolf, Ph.D. Thesis, Technical University of Braunschweig (2003).
		\bibitem{toth} G. Toth and O. G\"uhne. Phys. Rev. Lett. {\bf 102}, 170503 (2009).
			\bibitem{weedbrook} C. Weedbrook, S. Pirandola, R. García-Patrón, N. J. Cerf, T. C. Ralph, J. H. Shapiro, and S. Lloyd
		Rev. Mod. Phys. \textbf{84}, 621 (2012).
		

		\bibitem{Anaellethesis} A. Hertz. {\it Exploring continuous-variable entropic uncertainty relations and separability criteria in quantum phase space}. Ph.D thesis Université libre de Bruxelles (2018).
		
		\bibitem{fabre}C. Fabre and N. Treps
		Rev. Mod. Phys. \textbf{92}, 035005 (2020).
		\bibitem{HNLA} T. C. Ralph and A. P. Lund, in {\it Proc. of 9th Int. Conf. on Quantum Communication, Measurement, and Computing}, edited by
		A. Lvovsky, (AIP, New York, 2009), pp. 155–160.
		

		


\bibitem{universal-squeezer} C. N. Gagatsos, E. Karpov, and N. J. Cerf, Phys. Rev. A {\bf 86}, 012324 (2012).

\bibitem{He} M. He, R. Malaney, and B. A. Burnett
Phys. Rev. A \textbf{103}, 012414 (2021).
		
		\bibitem{gagatsos} C. N. Gagatsos, J. J. Fiur\'{a}\v{s}ek., A. Zavatta, M. Bellini and N. J. Cerf. Phys. Rev. A {\bf 89}, 062311 (2014).
		
		\bibitem{fiurasek2}J. Fiur\'{a}\v{s}ek. and N. J. Cerf. Phys. Rev. A {\bf 86}, 060302(R) (2012).

\bibitem{HNLAtt} M. Mi\ifmmode \check{c}\else \v{c}\fi{}uda, I. Straka, M. Mikov\'a, M. Du\ifmmode \check{s}\else \v{s}\fi{}ek, N. J. Cerf, J. Fiur\'a\ifmmode \check{s}\else \v{s}\fi{}ek, and M. Je\ifmmode \check{z}\else \v{z}\fi{}ek, Phys. Rev. Lett. {\bf 109}, 180503 (2012).

		
		\bibitem{holevo}A. S. Holevo. Probl.
Inf. Transm., {\bf 44} 171 (2008).




				\bibitem{fiurasek}J. Fiur\'{a}\v{s}ek.  Phy. Rev. Lett., {\bf 89}(13):137904  (2002).

\bibitem{Giedke}G. Giedke and J. I. Cirac. Phys. Rev. A, {\bf66}(3):032316 (2002).



\end{thebibliography}
\end{document}